\documentclass[12pt]{amsart}  
\usepackage{geometry} 
\usepackage{graphicx}
\usepackage{amsfonts}
\usepackage{amsmath}
\usepackage{amssymb}
\usepackage[arrow,matrix,curve]{xy}
\usepackage{enumerate}
\usepackage{color}
\usepackage[normalem]{ulem}
\usepackage{tikz}
\usetikzlibrary{decorations.pathmorphing}
\usepackage{xspace}
\definecolor{light-blue}{rgb}{0.8,0.85,1}
\definecolor{light-red}{rgb}{1,.4,.4}
\definecolor{purp}{rgb}{.7,.3,1}
\definecolor{yel}{rgb}{1,1,.5}
\definecolor{cy}{rgb}{0,1,1}
\usepackage{colortbl}
\usepackage{multirow}
\usepackage{array}

\newtheorem{proposition}{Proposition}
\theoremstyle{definition}

\newcommand{\co}{\colon\,}
\newcommand{\bT}{\mathbb T}
\newcommand{\bR}{\mathbb R}
\newcommand{\bC}{\mathbb C}
\newcommand{\bZ}{\mathbb Z}

\newcommand{\cK}{\mathcal K}
\newcommand{\cS}{\mathcal S}

\newcommand{\SL}{\mathop{\rm SL}}
\newcommand{\PSL}{\mathop{\rm PSL}}

\newcommand{\pt}{\text{pt}}
\newcommand{\lp}{\textup{(}}
\newcommand{\rp}{\textup{)}}

\newcommand{\im}{\operatorname{Im}}
\newcommand{\re}{\operatorname{Re}}
\newcommand{\tKR}{\widetilde{KR}}
\renewcommand\Re{\operatorname{Re}}

\newcommand{\BrR}{\operatorname{BrR}}
\newcommand{\Inv}{\operatorname{Inv}}
\newcommand{\sn}{\operatorname{sn}}
\newcommand{\cn}{\operatorname{cn}}
\newcommand{\dn}{\operatorname{dn}}
\newcommand{\jsc}{\operatorname{sc}}
\newcommand{\sd}{\operatorname{sd}}
\title[Elliptic curve orientifolds and $KR$]{String theory on 
elliptic curve orientifolds\\ and $KR$-theory}

\author{Charles Doran \and Stefan M\'{e}ndez-Diez}
\address{Department of Mathematical and Statistical Sciences,
University of Alberta,
Edmonton, AB T6G 2G1, Canada}
\curraddr[Stefan Mendez-Diez]{Department of Mathematics \& Statistics,
Utah State University,
Logan, UT 84322-3900, USA}
\email[Charles Doran]{doran@math.ualberta.ca}
\email[Stefan Mendez-Diez]{sdiez@math.ualberta.ca}
\thanks{CD and SMD supported by the
  Natural Sciences and Engineering Resource 
  Council of Canada, the Pacific Institute for the Mathematical
  Sciences, and the McCalla Professorship at the University of Alberta.}
  
\author{Jonathan Rosenberg}
\address{Department of Mathematics,
University of Maryland,
College Park, MD 20742-4015, USA} 
\email[Jonathan Rosenberg]{jmr@math.umd.edu}
\thanks{JR partially supported by NSF grant DMS-1206159.}

\begin{document}

\begin{abstract}
We analyze the brane content and charges in
all of the orientifold string theories on space-times of the
form $E\times \bR^8$, where $E$ is an elliptic curve with
holomorphic or anti-holomorphic involution. Many of these theories
involve ``twistings'' coming from the $B$-field and/or sign choices
on the orientifold planes. A description of these
theories from the point of view of algebraic geometry, using the
Legendre normal form, naturally divides them
into three groupings. The physical theories within each grouping
are related to one another via sequences of $T$-dualities. Our approach
agrees with both previous topological calculations of twisted
$KR$-theory and known physics arguments, and
explains how the twistings originate from both a mathematical and a
physical perspective. 
\end{abstract}

\maketitle
\tableofcontents
\section{Introduction}
\label{sec:intro}

The purpose of this paper is to study type IIA and IIB string theories
on all possible orientifold backgrounds for which the underlying
spacetime manifold $X$ is $\bT^2 \times \bR^8$. The $\bT^2$ factor should
be equipped with a complex structure, making it into an elliptic
curve (over $\bC$), as well as with a holomorphic or anti-holomorphic
involution $\iota$, which defines the orientifold
structure.\footnote{Note that in some of the literature, the word 
  ``orientifold'' is used to denote the quotient space $X/\iota$, but
  it is really essential to keep track of the \emph{pair} $(X,\iota)$
  and not just the quotient.}  (We extend the involution
$\iota$ to $X$ by making it trivial on the $\bR^8$ factor.) 
We discover that, and also explain why, the orientifold theories on
elliptic curves are naturally divided into three groupings, with
the theories in each grouping related to one another by sequences of
$T$-dualities. 

This is quite a natural problem for a variety of reasons. Compactifying
string theories
on elliptic curves is motivated by the fact that they are the simplest
compact Calabi-Yau manifolds (complex manifolds with a global
non-vanishing holomorphic volume form). Working with orientifolds is
natural also --- the orientifold construction generalizes the GSO
(Gliozzi-Scherk-Olive)
projection and encompasses most of the standard supersymmetric string theories.

\subsection{Motivation}
\label{sec:motivation}
The sigma-model of orientifold string theory on a spacetime $X$ with
involution $\iota$
involves \emph{equivariant} maps $\varphi\co \Sigma\to X$, so
that $\iota\circ\varphi = \varphi\circ\Omega$. Here $\Sigma$ is an
oriented $2$-manifold, possibly with boundary (the case of open strings),
called the \emph{string worldsheet}, and $\Omega$, called the
\emph{worldsheet parity operator}, is an orientation-reversing
involution on $\Sigma$. We require $\Sigma/\Omega$, though not
necessarily $\Sigma$ itself, to be connected.  (Thus an allowable
possibility is $\Sigma = \Sigma_0 \amalg \overline{\Sigma}_0$, where
$\Sigma_0$ is a connected oriented surface, $\overline{\Sigma}_0$ is the same
surface with orientation reversed, and $\Omega$ interchanges the two.)
See for example \cite{Distler:2009ri}; there some extra twisting data,
which we are ignoring for the moment, is also taken into account, and
the notation is slightly different.

Orientifold string theories include all of the standard theories of
types IIA, IIB, and I, as well as a number of variants sometimes
denoted IA, \~I and $\widetilde{\mathrm IA}$. We analyze all possible
$T$-duality relationships between these theories when $X$ is the
product of an elliptic curve with flat $8$-space. It will be apparent from
the results below that all of these theories should be considered 
together, since they are linked to one another by $T$-duality.
 
For orientifold theories, as explained in \cite[\S5.2]{Witten:1998},
\cite{Hori:1999me} and 
\cite{Gukov:1999}, D-brane charges are given by
$KR$-theory in the sense of Atiyah \cite{MR0206940}. We compute
the relevant $KR$-groups in all cases, and relate these groups to the
actual branes that arise. We also study how the $KR$-groups and branes
are related under $T$-duality and mirror symmetry. In this context it is
useful to quote from \cite[\S6]{MR2116734}: ``Since $T$-duality is
related to the Fourier transform, and since the Fourier 
transform of a real function is not necessarily real, a theory of
$T$-duality in type I string theory necessarily involves $KR$-theory, or
Real $K$-theory in the sense of Atiyah.'' 

There is already
a fair body of literature on orientifold compactifications on $S^1$,
and there is even some literature on $T^2$ orientifolds (e.g.,
\cite[\S7.2]{Gao:2010ava}). However, to our knowledge, this is the
first attempt at a systematic study of all type II orientifold string
theories on $\bT^2\times \bR^8$ that includes a calculation of all the
$KR$ groups and a study of all possible $T$-dualities. We also take into
account all possibilities for the complex structure, using the
classification in \cite{Bates:2006}. Considering the complex structure
is important, since elliptic curves are the simplest case for checking
predictions of mirror symmetry.  Understanding elliptic curve
orientifolds will also be the first step in attacking orientifolds on
higher-dimensional Calabi-Yau manifolds, such as abelian varieties, K3
surfaces, and most of all, Calabi-Yau 3-folds.  For example, a large
class of interesting K3 surfaces come with elliptic curve fibrations.

\subsection{Outline of the paper}
\label{sec:outline}

This paper begins in Section \ref{sec:holoandantiholo}
with a review of the classification of holomorphic
and anti-holomorphic involutions on elliptic curves, taken from
\cite{Bates:2006}. The classification of anti-holomorphic involutions
is equivalent to the classification of elliptic curves defined over
$\bR$, found in \cite{MR640091}. Next, in Section
\ref{sec:KRcalcs} we review the $KR$-theory of Atiyah and all its twisted
versions (including those coming from a sign choice on the components
of the fixed set). We then record all the groups that occur for the
various possible involutions and twistings. Most of these calculations
are taken from \cite{Doran:2013sxa}, but we also 
relate the results to earlier calculations
made in \cite{MR1936583} and \cite{Karoubi:2005} and to
classifications of twistings of $KR$ by Moutuou
\cite{MoutuouThesis,2011arXiv1110.6836M}. 

The heart of this paper consists of Sections \ref{sec:Tduality},
\ref{sec:geoTduality}, and \ref{sec:branes}. We begin by describing
the $T$-dualities that relate the various orientifold string theories on
elliptic curves (with a holomorphic or anti-holomorphic
involution). Most of these theories only live on a certain portion of
the moduli space of elliptic curves with K\"ahler structure and
$B$-field.  This moduli space is described by two parameters $\tau$
(describing the complex structure) and $\rho$ (describing the K\"ahler
form and $B$-field), which are interchanged under $T$-duality.
It turns out that the orientifold theories break into three groupings,
and iterated $T$-dualities relate all of the theories in a single
grouping. This fact was known before (e.g., in \cite{Gao:2010ava},
though some cases go back to \cite{MR1075783}, \cite{Hori:1999me} and
\cite{Witten:1998-02}), but our description of what happens to the
involutions is more explicit. In Section \ref{sec:geoTduality}, we
attack the problem of how to explain the three $T$-duality groupings
in purely geometric terms, without recourse to physical arguments.
Here it turns out that algebraic and complex geometry plays a crucial
role; the $T$-duality groupings can be explained perfectly in terms of
the Legendre normal forms of real elliptic curves, and the
uniformization of these curves in terms of Jacobi elliptic
functions. Finally, in Section \ref{sec:branes}, we give a complete
description of the $D$-brane and $O$-plane charges in the various
theories, and explain how these transform under $T$-duality.

\section{The classification of holomorphic and anti-holomorphic involutions}
\label{sec:holoandantiholo}

A torus $\bT^2$ with a complex structure  can be identified with 
$\bC/\Lambda$
for some lattice $\Lambda$. The holomorphic maps $\bC/\Lambda \to
\bC/\Lambda'$ 
are given by complex affine maps $z\mapsto \gamma z
+ \delta$ sending $\Lambda$ into $\Lambda'$. Thus we can rotate and scale so
that the lattice $\Lambda$ is generated by $1$ and a complex number
$\tau$ with $\im{\tau}>0$. Note that
$$\tau\mapsto\frac{a\tau+b}{c\tau+d},\mbox{    }\left(\begin{array}{cc} a
  & b \\ c & d \end{array}\right)\in\PSL(2,\bZ)$$ 
leaves the torus invariant up to holomorphic isomorphism. For
applications to string theory we want 
our torus to be equipped with a K\"ahler form $J\sim\sqrt{G}dx\wedge
dy$ and the NS-NS $2$-form $B$-field $B$, which combine to give an
invariant $\rho=\int_{\bT^2}(B+iJ)$ in the upper half-plane.
$T$-duality along with the gauge
invariance $\rho\mapsto \rho+1$ implies 
$$\rho\mapsto\frac{a\rho+b}{c\rho+d},\mbox{    }\left(\begin{array}{cc} a
  & b \\ c & d \end{array}\right)\in\PSL(2,\bZ)$$ 
also leaves the torus invariant.  
Therefore, the quantum moduli space
of $\bT^2$ (with its geometry as given by $\rho$)
is given by a product of
two copies of the quotient of the upper half-plane by $\PSL(2,\bZ)$.
In this context, mirror symmetry \cite{MR1358624}
corresponds to the interchange $(\tau,\rho)\mapsto(\rho,\tau)$.

In \cite{Bates:2006}, the authors look at all holomorphic and
anti-holomorphic involutions 
of $\bT^2$ combined with the worldsheet parity operator, which
correspond to the possible orientifold structures for type IIB and
type IIA theories, respectively.

The fixed set of a holomorphic involution on a complex elliptic curve
$E$ is a closed complex
submanifold, hence is either empty, a finite non-empty set, or
everything.  Holomorphic involutions are always of the form $z\mapsto
\pm z + \delta$. When we choose the $+$ sign, $\delta$ is a
$2$-torsion point in $E$, hence is $0$ (giving the trivial involution)
or an element of $E$ of order precisely $2$ (giving a free
involution). When we choose the $-$ sign, $\delta$ can be any point in
$E$ and there are exactly $4$ fixed points (the $2$-torsion points in
$E$ shifted by $\delta/2$). 

An anti-holomorphic involution $\varphi$ of $\bT^2$ must be induced by
a self-map $z\mapsto\alpha\bar z+\beta$ of $\bC$ preserving $\Lambda$ and
of order $2$ modulo translation by an element of $\Lambda$. 
All of the anti-holomorphic involutions of $\bT^2$ were worked out in \cite{MR640091} and they are given by Table
\ref{Table:class_invo}.  A necessary and sufficient condition for an
elliptic curve to admit an anti-holomorphic involution is for its
$j$-invariant to be real, $j(\tau)\in \bR$. 

\begin{table}[htb]
\noindent\makebox[\textwidth]{\begin{tabular}{|| c | c | c | c | c | c | c ||}
\hline
Case & $\tau$ & $j(\tau)$ & $\alpha$ & $\beta$ & $s$ & Fixed pts\\
\hline\hline
\multirow{4}{*}{(a)} & \multirow{4}{*}{\small $i\tau_2$ with $\tau_2>1$} &
\multirow{4}{*}{$j>1$} & $1$ & $0$ &$2$ & $\im(z)=0$;
$\im(z)=\tau_2/2$ \\ \cline{4-7} 
& & & $-1$ & $0$ &$2$ & $\re(z)=0$; $\re(z)=1/2$ \\ \cline{4-7}
& & & $1$ & $1/2$ &$0$ &  \\ \cline{4-7}
& & & $-1$ & $\tau/2$ &$0$ & \\ \hline
\multirow{3}{*}{(b)} & \multirow{3}{*}{$i$} & \multirow{3}{*}{$1$} &
$1\sim -1$ & $0$ &$2$ & $\im(z)=0$; $\im(z)=1/2$ \\ \cline{4-7} 
& & & $i\sim -i$ & $0$ &$1$ & $z=re^{i\pi/4}$, $r\in\bR$ \\ \cline{4-7}
& & & $1\sim -1$ & $1/2$ &$0$ &  \\ \hline
\multirow{2}{*}{(c)} & \multirow{2}{*}{\small $e^{i\theta}$ with
  $\pi/3<\theta<\pi/2$} & \multirow{2}{*}{$(0,1)$} & $\tau$ & $0$ &$1$
& $z=re^{i\theta/2}$, $r\in\bR$ \\ \cline{4-7} 
& & & $-\tau$ & $0$ &$1$ & $z=ire^{i\theta/2}$, $r\in\bR$ \\ \hline
\multirow{2}{*}{(d)} & \multirow{2}{*}{$e^{i\pi/3}$} &
\multirow{2}{*}{$0$} & $1\sim e^{2i\pi/3}\sim e^{4i\pi/2}$ & $0$ &$1$
& $\im(z)=0,\sqrt{3}/2$ \\ \cline{4-7} 
& & & $e^{i\pi/3}\sim -1\sim e^{5i\pi/3}$ & $0$ &$1$ & $\re(z)=0,1/2$ \\ \hline
\multirow{2}{*}{(e)} & \multirow{2}{*}{\small $\frac{1}{2}+i\tau_2$ with
  $\tau_2>\frac{1}{2}\sqrt{3}$} & \multirow{2}{*}{$j<0$} & $1$ & $0$
&$1$ & $\im(z)=0,\tau_2$ \\ \cline{4-7} 
& & & $-1$ & $0$ &$1$ & $\re(z)=0,1/2$ \\ \hline
\end{tabular}}
\smallskip
\caption{Table of anti-holomorphic involutions}
\label{Table:class_invo}
\end{table}

Table \ref{Table:class_invo} gives the invariant known as the 
\textit{species}, $s$, of each involution. The species gives the number of
components of the fixed point locus of the involution. The authors of
\cite{Bates:2006} show that the species also gives the charges of the
$O$-planes present. The classification in Table \ref{Table:class_invo}
also has an interpretation in terms of algebraic geometry. Any complex
torus of complex dimension $1$ is automatically a smooth projective
variety and an elliptic curve $E$ defined over $\bC$. An
anti-holomorphic involution $\iota$ makes this into a \emph{real} elliptic
curve; i.e., $E$ is defined over $\bR$ and $\iota$ corresponds to the
action of $\text{Gal}(\bC/\bR)$ on $E(\bC)$. The fixed set $E^\iota$ is the
set of real points $E(\bR)$; topologically it is just a disjoint union
of $s$ circles.  The fact that $s\le 2$ is just a special case (since
elliptic curves have genus $1$) of Harnack's curve theorem, and the
classification by species is familiar from the theory of real elliptic
curves \cite{MR640091}. The classification of IIA orientifold theories
by species was pointed out by Sagnotti in \cite{MR1008294}.

As we said earlier, when we combine the involutions in Table
\ref{Table:class_invo} with the worldsheet parity operator, $D$-brane
and $O$-plane charges should be classified by $KR$-theory. In
\cite[Example A.5]{Karoubi:2005} the authors calculate the $KR$-theory for
involutions with non-trivial species, i.e., $s=1$ or $2$. They show 
$$\begin{aligned}
KR^0(\bT^2)&\cong\bZ^2\oplus\bZ_2^{s-1},\\
KR^{-1}(\bT^2)&\cong\bZ\oplus\bZ_2^{s+1}.
\end{aligned}$$

Note that when the fixed locus has $2$ components,
$KR^{-1}(\bT^2\times\bR^8)$ is isomorphic to
$KO(\bT^2\times\bR^8)$. At the moment, this might appear accidental,
but we will 
see that this can be explained by a chain of $T$-duality isomorphisms.

\section{$KR$ with a sign choice and calculations for tori}
\label{sec:KRcalcs}

It was proposed in \cite{Minasian:1997} and \cite{Witten:1998}, 
and is now generally accepted, that D-brane charges in string theory
should be classified by some variant of $K$-theory. In 
\emph{orientifold} theories, charges should be classified by some variant
of $KR$-theory, as described by Witten in \cite{Witten:1998}. 
However, classical $KR$-theory can only apply when all $O$-planes
have the same charge.  When $O$-planes with opposite charges are
present, the appropriate substitute is \emph{$KR$-theory with a sign
  choice}, which we described in the companion paper
\cite{Doran:2013sxa}. In this section, we will briefly review
$KR$-theory and $KR$-theory with a sign choice, as well as certain
twisted variants.  All these twistings of $KR$-theory were discussed
and classified by Moutuou 
\cite{MoutuouThesis,2011arXiv1110.6836M,2012arXiv1202.2057M},
though this may not be readily apparent because of the great
generality of Moutuou's framework. (Moutuou deals with $\bZ_2$-graded
algebras over
Real groupoids, but here we only need the case where the grading is
trivial and the groupoid reduces to a Real space.)

We will also discuss some of the different notations
appearing in the literature and the relations between them, and review
and further amplify the calculations from \cite{Doran:2013sxa} for the
case of $2$-torus orientifolds.  This section is purely topological;
we temporarily ignore geometrical structures such as Riemannian
metrics, complex structures, and K{\"a}hler forms, except insofar as
they illuminate the topology.

$KR$-theory, in the sense of Atiyah \cite{MR0206940}, is the cohomology
theory that classifies stable isomorphism classes of virtual Real vector
bundles on a ``Real'' space $(X,\iota)$. A Real space is a locally
compact (Hausdorff) space $X$, together with 
a self-homeomorphism $\iota$ of $X$ of period $2$. A Real vector
bundle on such a space is a complex vector bundle $E$, together with a
conjugate-linear bundle automorphism of $E$ of period $2$, covering $\iota$.
If $X$ is compact, $KR(X)$ is the group of formal differences
$[E]-[F]$, where $E$ and $F$ are Real vector bundles over $X$, and we
identify $[E]-[F]$ with $[E']-[F']$ if there is an isomorphism of Real
bundles $E\oplus F'\oplus G \cong E'\oplus F\oplus G$ for some Real
bundle $G$ over $X$. When $X$ is only \emph{locally} compact, $KR(X)$
is defined similarly, but with $E$ and $F$ required to be trivialized
and isomorphic in a neighborhood of infinity.

For string theory on a smooth manifold $X$, the charges of
$D$-branes are classified by 
pairs of vector bundles $(E,F)$, the Chan-Paton bundles on the branes
and anti-branes, modulo the equivalence $(E,F)\sim(E\oplus H, F\oplus
H)$. $D$-branes on orientifolds of the form $X/(\iota\cdot\Omega)$,
where $X$ is a smooth manifold, $\iota$ is an involution on $X$, and
$\Omega$ is the world sheet parity operator, are classified by vector
bundles on $X$ that are equivariant under the action of
$\iota\cdot\Omega$. $\Omega$ sends a vector bundle $E$ to its complex
conjugate $\bar{E}$. Therefore, a vector bundle $E$ is
$\iota\cdot\Omega$-equivariant if there exists an isomorphism,
$\varphi$, from the pullback $\iota^*E$ to $\bar{E}$ such that
$(\varphi\iota^*)^2=1$, which is exactly the Reality condition of
Atiyah. Thus we naturally arrive at 
the group $KR(X)$ (the spacetime involution $\iota$ being understood). 

More generally, $D$-brane charges are classified by $KR^{-j}(X)$, where the
index $j$ depends on the dimension of the brane. To define
the higher $KR$-groups we must first introduce some notation. Let
$\bR^{p,q}=\bR^p+i\bR^q$ with the involution $\iota$ given by complex
conjugation, and let $S^{p,q}$ be the unit sphere (of
dimension $p+q-1$) in $\bR^{p,q}$. (In
this notation, the roles of $p$ and $q$ are the reverse of those in the
notation used by Atiyah in \cite{MR0206940}, but the same as the
notation in \cite{MR1031992}, \cite{Bergman:1999} and
\cite{Olsen:1999}.) We define 
$$KR^{p,q}(X)=KR(X\times\bR^{p,q}).$$
This obeys the periodicity condition \cite[Theorem 2.3]{MR0206940}
$$KR^{p,q}(X)\cong KR^{p+1,q+1}(X),$$
where the isomorphism is given by cup product with the Bott class.
Since $KR^{p,q}$ only depends on the difference $p-q$, we can define
$$KR^{q-p}(X)=KR^{p,q}(X).$$
$KR^j(X)$ is periodic with period $8$.

When we compactify string theory on on an $m$-dimensional
space $M$, so that the spacetime manifold is $\bR^{10-m,0}\times M$,
we are interested in the charges of $D$-branes in the non-compact
dimensions. So we want to consider $Dp$-branes of codimension $9-m-p$
in $\bR^{9-m,0}$. These can arise from both $Dp$-branes located at a
particular point in $M$ or higher dimensional $D$-branes that wrap
non-trivial cycles in M. Furthermore, we only want to consider
systems with finite energy, so we only want to classify systems that are
asymptotically equivalent to the vacuum in the transverse space
$\bR^{9-m-p,0}$. That means that the system must be equivalent to the
vacuum on an entire copy of $M$ at infinity. Mathematically this means
we want to add a copy of $M$ at infinity (i.e., take the 
product with $M$ of the one-point
compactification of $\bR^{9-m-p,0}$) and consider bundles on
$S^{10-m-p,0}\times M$ that are trivialized on the copy of $M$ at
infinity. Such bundles are classified by $KR^{-i}(S^{10-m-p,0}\times
M, M)$. This can be related to the $KR$-theory of $M$ via the
isomorphism 
\begin{equation}
\label{eqn:relvabs}
KR^{-i}(S^{10-m-p,0}\times M, M)\cong KR^{p+m-9-i}(M).
\end{equation}
$Dp$-brane charges are classified by $KR^{p+m-9-i}(M)$ where $i$ will
depend on the string theory and $M$. We are considering the case when
$M$ is a Real elliptic curve, so $m=2$.  

This classification of $D$-brane charges includes the usual
classification of type I brane charges by $KO$-theory and type II
brane charges by complex $K$-theory. The type I theory is obtained by letting
$\iota=1$. This corresponds to the well known fact that the type I
theory is
the type IIB theory divided out by the action of $\Omega$. In terms of
the $KR$-theory classification, being equivariant means that $E$ is
isomorphic (in a way fixing the base $X$)
to $\bar{E}$, or that $E$ is real. The classification of
equivariant Real bundles on $X$ is thus the same as that of real
bundles on $X$, giving the well known mathematical result \cite{MR0206940} 
$$KR(X)\cong KO(X)$$
when $\iota$ is trivial. To obtain the usual type II classification of
$D$-branes in a spacetime $X$, we use the result from
\cite[Proposition 3.3]{MR0206940} 
$$KR\left(\xymatrix@C=0pt@M-1pt{X \ar@{<->}@/^1pc/[rr]&\amalg&X}
  \right)\cong K(X),$$
where the involution exchanges the two copies of $X$.

Often, when studying the $K$-theory classification of $D$-branes for
the type II theories on a smooth manifold, the full indexing of
$K^{-i}(X)$ is ignored, since it has period $2$. While this is often
most convenient for the purposes of mathematical calculations, to
determine the brane content it is often more useful to use the
relative $K$-theory given by the isomorphism \eqref{eqn:relvabs}. For
the trivial case of type IIB $D$-branes in Minkowski spacetime, the
distinction between $K^0(\pt)$ to classify $D9$-branes and
$K^{-2}(\pt)$ is inconsequential. However, for our current purposes,
the distinction is very important. So we will want to keep track of
the full $\bZ/8$-graded group $KR^*(X)$.

$KR$-theory with a sign choice, introduced in \cite{Doran:2013sxa}, is
a variant of $KR$-theory for a Real space $(X,\iota)$ with a choice 
$\alpha$ of $\pm$
signs, one for each component of the fixed set $X^\iota$. This theory
needs to be defined via noncommutative geometry, and we refer the
reader to \cite{Doran:2013sxa} for the precise definition, but it has
the property that on a component $F$ of $X^\iota$ with positive sign
choice, $KR_\alpha^*(F) = KO^*(F)$, the usual $K$-theory of real
vector bundles, whereas on a component $F$ of $X^\iota$ with negative sign
choice, $KR_\alpha^*(F) = KSp^*(F)$, the $K$-theory of quaternionic
vector bundles.  This is precisely what is appropriate if $F$ is an
$O^+$- (resp., $O^-$-) plane. (Note that there is some
disagreement in the literature about what should be called an $O^+$-plane
and what should be called an $O^-$-plane, but we are following the
convention in \cite[\S2.3]{Witten:1998-02}.  As Witten points out, the
associated tadpoles have \emph{opposite} sign.)

The basic facts about $KR$-theory 
can be found in \cite{MR0206940} or in \cite[\S1.10]{MR1031992} ---
note that these sources use opposite indexing conventions and that we are
following Lawson-Michelsohn, not Atiyah, so that $\bR^{p,q}
=\bR^p\oplus i\bR^q$ with involution fixing the $\bR^p$ summand 
and multiplying by $-1$ on the $\bR^q$ summand. For locally compact but
non-compact Real spaces, we always use $KR$-theory with compact supports. 
For any real space $(X,\iota)$ (often we will suppress the involution
in the notation), $KR^j(\bR^{1,0}\times X) \cong KR^{j-1}(X)$
and $KR^j(\bR^{0,1}\times X) \cong KR^{j+1}(X)$. If $X$ is compact and
has an $\iota$-fixed point $x_0$, then the inclusion
$\{x_0\}\hookrightarrow X$ is equivariant and equivariantly split, so
$KR^j(X) \cong \tKR^j(X) \oplus KO^j$, where $KO^j$ means $KO^j(\pt)$
and $\tKR^j(X) = KR^j(X\smallsetminus \{x_0\})$. Thus
$KR^j(S^{1,1})\cong  KR^j(\bR^{0,1}) \oplus KO^j \cong KO^{j+1} \oplus KO^j$, and
$KR^j(S^{2,0}) \cong  KR^j(\bR^{1,0}) \oplus KO^j \cong KO^{j-1} \oplus KO^j$. We
also have $KR^j(S^{0,2}\times X) \cong KSC^j(X)$, the self-conjugate
$K$-theory of Anderson \cite{Anderson}
and Green \cite{MR0164347}, by \cite[Proposition 3.5]{MR0206940}.

Note that since our spacetime manifolds will always be of the form
$X\times \bR^{8,0}$, where $X$ is a two-dimensional Real space, and
since $KR$-theory has Bott periodicity of period $8$, there is a
natural isomorphism $KR^j(X\times \bR^{8,0}) \cong KR^j(X)$, and we
can ignore the $\bR^8$ factor for purposes of this section. (However,
it will be needed in Section \ref{sec:Tduality} when we talk about
specific branes.)

In \cite{Doran:2013sxa}, we computed the $KR$ with a sign choice for
all possible holomorphic or antiholomorphic involutions on complex
elliptic curves $X$. In fact there are not that many different
topological possibilities. 

\subsection{Holomorphic involutions}
\label{sec:KRholo}

If the involution is holomorphic, either it
is trivial, or $X$ is homeomorphic to $S^{1,1}\times
S^{1,1}$ as a Real space, or the involution is free and $X$ is homeomorphic to
$S^{0,2}\times S^{2,0}$. 

\subsubsection{Trivial involutions}
\label{sec:KRtriv}

For spaces with trivial involution, $KR$-theory reduces to
$KO$-theory. Topologically, a $\bT^2$ with trivial involution is just
the Real space $S^{2,0}\times S^{2,0}$, and $KR^j(S^{2,0}\times
S^{2,0}) \cong KO^j(S^1\times S^1) \cong KO^j(S^1) \oplus
KO^{j-1}(S^1) \cong KO^j \oplus KO^{j-1} \oplus KO^{j-1} \oplus
KO^{j-2}$. The associated physical theory is the type I string theory on $\bT^2$.

Just for completeness, note that if $E$ is an elliptic curve with
trivial involution, we can put a holomorphic involution on $E\amalg E$
that simply interchanges the two factors. This space is $\bT^2\times
S^{0,1}$ as a Real space, and $KR^j(\bT^2\times S^{0,1}) \cong
K^j(\bT^2)$, which is $\cong \bZ^2$ in each degree.  The associated
physical theory is ordinary Type IIB theory on $E$ (with no
involution). 

\subsubsection{Four fixed points}
\label{sec:KR4fixed}

A $\bT^2$ with a holomorphic involution with four fixed points is
topologically just $S^{1,1}\times S^{1,1}$. And we obtain
\[
\begin{aligned}
KR^j(S^{1,1}\times
S^{1,1}) &\cong KR^j(S^{1,1}) \oplus KR^j(\bR^{0,1}\times S^{1,1})\\
&\cong KR^j(S^{1,1}) \oplus KR^{j+1}(S^{1,1})\\
&\cong KO^j \oplus KO^{j+1} \oplus KO^{j+1} \oplus KO^{j+2}.
\end{aligned}
\]

When there are four fixed points, there are two other interesting
possible assignments of signs. When the sign choice is $(+,+,-,-)$, we
can identify $X$ with 
\[
S^{1,1}_{(+,-)}\times S^{1,1}\cong (S^{1,1}_{(+,-)}\times \{\pt\}) \amalg
(S^{1,1}_{(+,-)}\times \bR^{0,1})
\]
and we obtain
\[
KR^j_{(+,+,-,-)}(S^{1,1}\times S^{1,1})\cong KSC^{j+2} \oplus KSC^{j+1},
\]
as was shown in \cite{Doran:2013sxa}.

The sign choice $(+,+,+,-)$ requires a more complicated calculation
which was done in \cite{Doran:2013sxa}; the result appears in Table
\ref{table:twistedKO}.

\begin{table}[htb]
\begin{center}
\begin{tabular}{||c||c||c||c||}
\hline
$j$ mod $8$ & $KO^j(T^2,\widetilde w_2)$ & $KR^j(\text{species }1)$
& $KR^j_{(+,+,+,-)}(S^{1,1}\times S^{1,1})$\\
\hline\hline
$0$ &$\bZ\oplus\bZ_2^2$ &$\bZ^2$ & $\bZ$\\
$-1$&$\bZ_2^2$ &$\bZ\oplus\bZ_2^2$ & $\bZ^2$ \\
$-2$&$\bZ$ &$\bZ_2^2$ & $\bZ\oplus\bZ_2^2$\\
$-3$&$\bZ^2$ &$\bZ$ & $\bZ_2^2$\\
$-4$&$\bZ$ &$\bZ^2$& $\bZ$ \\
$-5$&$0$ &$\bZ$ & $\bZ^2$ \\
$-6$&$\bZ$ &$0$ & $\bZ$ \\
$-7$&$\bZ^2$ &$\bZ$  & $0$\\
\hline\hline
\end{tabular}
\caption{$KO^j(T^2,\widetilde w_2)$, $KR^j(\text{species }1)$, and
  $KR^j_{(+,+,+,-)}(S^{1,1}\times S^{1,1})$}  
\label{table:twistedKO}
\end{center}
\end{table}

Note that mathematically we could also consider the sign choice
$(-,-,-,+)$. This however does not make physical sense. If the net
O-plane charge is negative, then tadpole cancellation would require
adding anti-branes, which would violate supersymmetry. For mathematical
completeness, we note that the relevant $KR$-groups can be obtained from
$KR_{(+,+,+,-)}^j(S^{1,1}\times S^{1,1})$ by shifting the index by
$4$. 

\subsubsection{Free involutions}
\label{sec:KRfreeholo}

A $\bT^2$ with a holomorphic involution with no fixed points is
topologically just $S^{0,2}\times S^{2,0}$ (which is equivalent to $S^{0,2}\times S^{0,2}$ as will be discussed later). And we obtain 
\[
\begin{aligned}
KR^j(S^{0,2}\times
S^{2,0}) &\cong KSC^j(S^1)\\
&\cong KSC^j \oplus KSC^{j-1}.
\end{aligned}
\]
Note that in this case the groups are periodic with period $4$, which
is in accordance with \cite[Proposition 1.8]{Karoubi:2005}, though in
general that statement is false ($S^{0,4}$ provides a counterexample,
as one can see from \cite{MR0206940}).

\subsection{Antiholomorphic involutions}
\label{sec:KRantiholo}

The study of $KR$-theory for antiholomorphic involutions is a special
case of the study of $KR$-theory for real algebraic curves. This has been
studied extensively in \cite{Karoubi:2005} and \cite{MR1936583},
which provide methods of calculation, though we will need to correct
two misprints in those papers. We can also take the antiholomorphic
involution on 
$E\amalg \bar E$ that interchanges the two factors, and we again get
complex $K$-theory $K^j(\bT^2)$, but this time with a focus on odd-dimensional
D-branes. The associated physical theory is ordinary Type IIA theory
on $E$ (with no involution).

\subsubsection{Species $2$}
\label{sec:KRspecies2}

A $\bT^2$ with an antiholomorphic involution of species $2$ is
topologically just $S^{1,1}\times S^{2,0}$. And we obtain 
\[
\begin{aligned}
KR^j(S^{1,1}\times S^{2,0}) &\cong KO^j(S^1) \oplus KO^{j+1}(S^1)\\
&\cong KO^j \oplus KO^{j-1}  \oplus KO^{j+1} \oplus  KO^j.
\end{aligned}
\]

In case of species $2$, there is also the sign choice $(+,-)$, in
which case we obtain 
\[
\begin{aligned}
KR^j_{(+,-)}(S^{1,1}\times S^{2,0}) &\cong KR^j_{(+,-)}(S^{1,1})
\oplus KR_{(+,-)}^{j-1}(S^{1,1})\\ 
&\cong KSC^{j+1} \oplus KSC^j.
\end{aligned}
\]

\subsubsection{Species $0$}
\label{sec:KRspecies0}

A $\bT^2$ with an antiholomorphic involution of species $0$ is
topologically just $S^{0,2}\times S^{1,1}$. And we obtain 
\[
\begin{aligned}
KR^j(S^{0,2}\times S^{1,1}) &\cong KSC^j(S^{1,1}) \\
&\cong KSC^j \oplus KSC^j(\bR^{0,1}) \cong  KSC^j \oplus KSC^{j+1}.
\end{aligned}
\]
Note that in this case the groups are periodic with period $4$.
Furthermore, the final result is in accordance with \cite[Example
  A.3]{Karoubi:2005} with genus $g=1$. (There is a small misprint in
  \cite{Karoubi:2005}; the calculation of $KR^{-*}(X)$ is correct and
  does follow from collapse of the spectral sequence $H^p_{G}(X; KR^q)
  \Rightarrow KR^{p+q}(X)$, but $E_2^{2,-2}=H^2(X/G; \bZ(-1))\cong
  \bZ$, not $0$. For purposes of our present application, $G=\bZ_2$
  and $X=\bT^2$, $X/G$ is a Klein bottle, and $H^2(X/G; \bZ(-1))\cong
  H_0(X/G; \bZ)\cong \bZ$ by (twisted) Poincar\'e duality.)

\subsubsection{Species $1$}
\label{sec:KRspecies1}

The calculation of $KR^j(X)$ when $X$ is a real elliptic curve of
species $1$ is a bit tricky and was done in \cite[Theorem
4]{Doran:2013sxa}. The result is that
\[
KR^j(X)\cong \left(KO^j\right)^2 \oplus K^{j-1},
\]
and also appears in Table \ref{table:twistedKO}.

It is interesting to compare this calculation with
\cite[Corollary 4.2]{Karoubi:2005}, that says that the natural map
$K_j(X;\bZ_2) \to KR^{-j}(X;\bZ_2)$ sending algebraic to topological
$K$-theory is an isomorphism for $j$ sufficiently large ($j\ge 1$ in
fact will do). Here $K$-theory or $KR$-theory with
$\bZ_2$ coefficients is related to the integral theory by a universal
coefficient or Bockstein exact sequence
\begin{equation}
\begin{aligned}
0 &\to KR^{-j}(X)/2 \to KR^{-j}(X;\bZ_2) \to {}_2KR^{-j+1}(X)\to 0,\\
0 &\to K_j(X)/2 \to KR_j(X; \bZ_2)  \to {}_2K_{j-1}(X)\to 0,
\end{aligned}
\label{eq:UCT}
\end{equation}
where ${}_2KR^{-j+1}(X)$ denotes the $2$-torsion in
$KR^{-j+1}(X)$, and similarly for $K_j$. The torsion subgroup
of $K_j(X)$ was computed in \cite{MR1936583}, but there
is a  small typo in the statement of \cite[Main Theorem
0.1]{MR1936583}. $K_2(X)_{\text{tors}}$ should contain $\nu+1$ copies of
$\bZ_2$ (here $\nu$ is the species), not $\nu$ copies as written. 
(This result was miscopied from \cite[Theorem 4.6]{MR1936583}.)
The $K$-theory with $\bZ_2$ coefficients, or the $KR$ theory with
$\bZ_2$ coefficients, is then as given in Table 
\ref{table:KZtwo}.

\begin{table}[htb]
\begin{center}
\begin{tabular}{||c||c||}
\hline
$j$ mod $8$ & order of $K_j(X;\bZ_2)\cong KR^{-j}(X;\bZ_2)$\\
\hline\hline
$0$&$2^2 $\\
$1$&$2^3 $ \\
$2$&$2^4 $\\
$3$&$2^3 $\\
$4$&$2^2 $ \\
$5$&$2 $ \\
$6$&$0 $ \\
$7$&$2 $\\
\hline\hline
\end{tabular}
\caption{algebraic $K$-theory mod $2$ for a real elliptic curve of
  species $1$}  
\label{table:KZtwo}
\end{center}
\end{table}

\subsection{Twisted groups}
\label{sec:KRtwist}

Finally, in the case of the trivial involution, we also have twisted
groups with a non-zero twist $\widetilde w_2 \in H^2(\bT^2, \bZ_2)$. Such
twisted $KO$-theory was introduced in \cite{MR0282363}, and can be
identified with the topological $K$-theory of a noncommutative algebra
that is locally, but not globally, isomorphic to continuous functions
on $\bT^2$ with values in a matrix algebra over $\bR$, since the
automorphism group of $M_n(\bR)$ has the homotopy type of $PO(n)$ and
$BPSO(2n)$ approximates $K(\bZ_2, 2)$ in low dimensions. The twisted
$KO$-groups also appear in Table \ref{table:twistedKO} and in Witten's
``theory with no vector structure'' \cite{Witten:1998-02}.

Twistings and sign choices in $KR$-theory
have been unified in work of {Moutuou} 
\cite{MoutuouThesis,2011arXiv1110.6836M}. He constructs and
computes a \emph{graded Brauer group} \cite{2012arXiv1202.2057M} of
graded real  continuous-trace algebras over
a Real space $(X,\iota)$. The equivalence relation is Morita
equivalence over $X$ and the group operation is graded tensor product
(over $X$). For our purposes we don't need the grading,
so we get a \emph{Brauer group} of (ungraded) real continuous-trace
algebras, which turns out to be
\begin{equation}
\BrR(X,\iota)\cong H^0(X^\iota, \bZ_2)\oplus H^2_\iota(X, \cS),
\label{eq:Mout}
\end{equation}
where the first summand is the group of \emph{sign choices} and the second
group is \emph{equivariant} sheaf cohomology (this is discussed in
greater detail in \cite{2012arXiv1202.0155M}) for the Real sheaf $\cS$
of germs of $S^1$-valued continuous functions and we use the complex
conjugation involution on $S^1$. The second summand encodes the
\emph{{\lp}Real{\rp} Dixmier-Douady class}; in the notation of
\cite[Proposition 4.4.9]{MoutuouThesis}, this is the ungraded analogue
$\BrR_0(X)$ of $\widehat{\BrR}_0(X)$.  In the same notation, the first
summand is $H^0(X, \Inv \cK)$, where  $\Inv \cK$ is the
ungraded analogue of the sheaf
$\Inv \widehat{\cK}$.  But it is easy to see that this sheaf is
supported on the fixed set, where it has stalks $\pm 1$ corresponding
to the two possible local sign choices (orthogonal type or symplectic
type), thus giving \eqref{eq:Mout}.

\section{$T$-duality}
\label{sec:Tduality}

In this section we will discuss how the various orientifolds
classified in Section \ref{sec:holoandantiholo} are related via
$T$-duality. These relationships have already been discussed in
\cite{Gao:2010ava}. In \cite{Doran:2013sxa} we showed that you need to
include a sign choice and that twisting $KR$-theory by this sign choice
correctly classifies charges in $T$-dual theories. While the need for
using $KR$-theory and the geometric meaning of the twisting caused by
the $B$-field are well understood, there is no purely geometric explanation
of why $T$-duality requires a sign choice. In this section we will
simply review the various $T$-duality relationships. We will give a
geometric description for all of the possible $T$-dualities between
elliptic curve orientifolds, including an explanation for all sources
of twisting (both sign choice and $B$-field) in the following
section. In Section \ref{sec:branes} we will describe the brane
content in the different theories using the $K$-theoretic analysis of
\cite{Doran:2013sxa} together with this geometric description. 

Since the right- and left-movers have the same chirality in the type
IIB theory, only holomorphic involutions are compatible with the type
IIB theory. Similarly the type IIA theory is only compatible with
antiholomorphic involutions since the left- and right-movers have
opposite chirality. Since $T$-duality (on a single circle factor)
interchanges the type IIA string theory with the type IIB theory, 
it also exchanges holomorphic and antiholomorphic involutions. 
The various theories can be broken into $3$ groups, with the 
theories in a single group related via $T$-duality. Note that real elliptic
curves (the spacetimes for type IIA orientifold theories)
are generally grouped by their species. However, as we saw in
\cite{Doran:2013sxa}, the type IA and $\widetilde{IA}$ theories are
both defined on species $2$ real elliptic curves but cannot be related
by a $T$-duality. Our geometric description will show that the type
$\widetilde{IA}$ theory should be grouped with the species $0$ real
elliptic curves even though it is species $2$. Since we have not yet
defined a mathematical way to define the different $T$-duality groupings,
we will classify them in terms of their physical theories for now. 

The first group contains the type I theory as well as as the 
type IIA theory on an annulus (known as the type I$'$ or IA theory) and the
type IIB theory on $\bT^2/\bZ_2$ with four fixed points (and four
$O7^+$-planes). The second group contains the type \~I and
type $\widetilde{IA}$ theories as well as the type IIA theory on a Klein
bottle and the type IIB theory with four fixed points corresponding to
$2$ $O7^-$-planes and $2$ $O7^+$-planes. The third group contains the
type I theory without vector structure described in
\cite{Witten:1998-02} (the type I theory with non-trivial $B$-field),
the type IIA theory on a M\"obius strip, and the type IIB theory with
$1$ $O7^-$-plane and $3$ $O7^+$-planes. The fact that the last two of
these theories belong in the same $T$-duality grouping was already
pointed out in \cite{MR1075783}.  Note that each of the $3$ groups
contains one type IIB theory with $4$ fixed points --- such theories
are classified the net $O$-plane charge --- and also contains one type
IIA theory on a 
quotient of the torus by an orientation-reversing involution.  This
would provide two natural ways to
classify the groups, but instead we choose to refer to them as the
type I, type \~I, and ``type I without vector structure'' groups. 

The physical moduli space for string theory on a real elliptic curve
is determined by the complex structure, $\tau$, and the K\"ahler modulus,
$\rho$. The moduli space for the type IIA theories (real elliptic
curves with an antiholomorphic involution) is shown in Figure
\ref{Fig:modulispace}.  
\begin{figure}[ht]
  \centering
  \mbox{\includegraphics[scale=0.5]{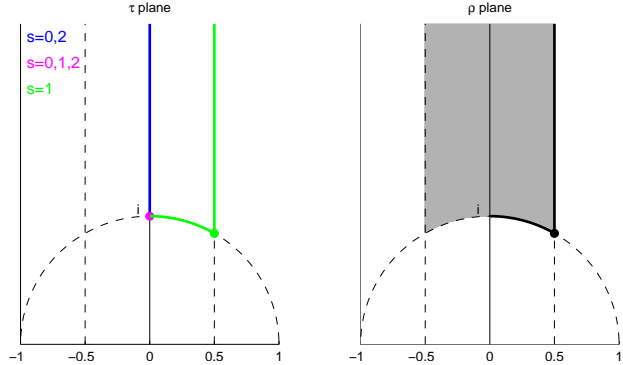}}
  \caption{Physical moduli space of string theory on a real elliptic
    curve with antiholomorphic involution corresponding to the type
    IIA theory.} 
\label{Fig:modulispace}
\end{figure}
This picture appeared already in \cite[Figs.\ 2 and 3]{MR1168626}.
As can be seen from the figure, the complex structure is constrained,
while the K\"ahler modulus is free. After a $T$-duality
transformation, we obtain the type IIB theory with the roles of $\tau$
and $\rho$ reversed. Therefore, the complex structure is free and the
K\"ahler modulus is constrained for holomorphic involutions. The
constraints on $\rho$ in the type IIB theory confirm the known result
that there are two possible values of the $B$-field in the type I theory,
$B=0,\frac{1}{2}$, corresponding to the $2$ vertical legs in the first
factor of Figure \ref{Fig:modulispace}. 

At first glance, the arc $\tau=e^{i\theta}$, with
$\frac{\pi}{3}<\theta<\frac{\pi}{2}$ would seem to imply the $T$-dual
IIB theory would have an unallowable value of $B$, since
$0<\Re{\rho}<\frac{1}{2}$. If we let $u=\sin{\theta}$ then the arc is
described by $\tau=\sqrt{1-u^2}+iu$, with
$\frac{\sqrt{3}}{2} < u < 1$. Performing the $\SL{(2,\bZ)}$
transformations $\tau\mapsto\tau-1$ and then
$\tau\mapsto-\frac{1}{\tau}$ sends $\tau$ to 
\begin{align*}
\widetilde{\tau}&=\frac{1}{1 - \tau}\\ 
&=\frac{1}{1-\sqrt{1-u^2}-iu}\\
&=\frac{1}{2}+i\frac{u}{2(1-\sqrt{1-u^2})}.
\end{align*}
Since $\tau$ and $\widetilde{\tau}$ are related by an $\SL{(2,\bZ)}$
transformation, they describe equivalent elliptic curves. This shows
us that for any real elliptic curve (elliptic curve with
antiholomorphic involution) there is a representative with
$\Re{\tau}=0$ or $\frac{1}{2}$. This matches with the fact that
that the only possible values of the $B$-field ($\Re{\rho}$) for type
IIB theories on elliptic curve orientifolds are $0$ and $\frac{1}{2}$. 

\begin{proposition}
\label{prop:IIAnoB}
The stable $D$-brane charges in type IIA orientifold theories on
elliptic curves do not not depend on the $B$-field. 
\end{proposition}
\begin{proof}
This follows immediately the observation in \cite{Bates:2006}
that a non-trivial $B$-field only affects $O$-planes that wrap the
entire elliptic curve. All of the type IIA theories contain either no
$O$-planes or $O$-planes that wrap a $1$-cycle in the elliptic
curve. Alternatively, to put this in purely mathematical terms,
twisting of $KR$-theory by the $B$-field amounts to a $\tilde w_2$
twist in $H^2$ (with $\bZ_2$ coefficients) of a component of the fixed
set of $\iota$, and for this to be non-zero requires $\iota\equiv 1$.

This can also be seen in the physical moduli space for type IIB
$\bT^2$ orientifolds versus the one for type IIA $\bT^2$ orientifolds
(Figure \ref{Fig:modulispace}). For type IIB theories, the K\"ahler
modulus, $\rho$, is constrained while the complex modulus, $\tau$, is
unconstrained. The opposite is true for the type IIA
theories. Changing the $B$-field for a type IIA elliptic curve
orientifold theory changes the complex structure of the $T$-dual IIB
elliptic curve orientifold theory. As can be seen in Figure
\ref{Fig:modulispace} and the list of holomorphic involutions given
at the beginning of Section \ref{sec:holoandantiholo}, all of the type
IIB elliptic curve orientifold theories are well defined for all
complex structures. Therefore the $D$-brane charges in the $T$-dual
IIB theory are independent of the choice of complex structure, and the
stable $D$-brane charges in the original IIA theory must not depend on
the B-field. 
\end{proof}

Proposition \ref{prop:IIAnoB} can be viewed as saying that the stable
$D$-brane charges in the type IIB elliptic curve orientifold theories
do not depend on whether or not you compactify on a rectangular or
a non-rectangular torus. It is important to note, however, that while
the $Dp$-brane \emph{charges} will remain 
unchanged, the actual \emph{sources} could be affected by a non-trivial
$B$-field. This is because a non-trivial $B$-field can affect
$D$-branes that wrap both compact directions
\cite{Bates:2006}. Another way to see this is that $T$-dualizing in a
direction that is not tangential or normal to the direction wrapped by
a brane will affect the resulting brane. We will find in Section
\ref{sec:branes} that the direction of $T$-duality is usually
constrained in the IIB theories to well defined directions relative to
the branes, making the brane content clear in a manner that is
independent of the choice of complex structure (or $B$-field for the
IIA theories). 

Using Proposition \ref{prop:IIAnoB} and the following discussion, we
will generally assume that all type IIA theories have trivial $B$-field, or
equivalently that all IIB theories are compactified on rectangular
tori. This assumption does not affect any of our results. We will
first consider the $2$ groupings containing only type IIB 
theories with trivial $B$-fields together, and then consider the
inclusion of non-trivial $B$-fields. We do this to clarify the
difference between twistings by the $B$-field and twistings by the
sign choice. 


\subsection{$T$-duality for elliptic curve orientifolds with trivial $B$-field}

Two of the three $T$-duality groups only contain type IIB theories
with trivial $B$-field. They are the group containing the type I
theory and the group containing the type \~I theory. These are the two
groups whose IIA theories only exist on rectangular tori. Let us first
consider the group containing the type I theory with trivial
$B$-field.  

\subsubsection{The type I theory}
The type I theory compactified on $\bT^2$ corresponds to the type IIB
orientifold theory compactified on $S^{2,0}\times S^{2,0}$. In
\cite{Doran:2013sxa} we described how the chain of $T$-dualities
starting from this theory can be obtained by compactifying the type
IIB theory on $S^{2,0}$ \cite{Olsen:1999,Bergman:1999} on an
additional circle.  

Beginning with the type IIB theory compactified on
$S^{2,0}\times S^{2,0}$, which is just the type I theory compactified
on $\bT^2\cong S^1\times S^1$, $T$-dualizing a single copy of
$S^{2,0}$ will transform it to 
$S^{1,1}$. Therefore, $T$-dualizing one circle of the type IIB theory
on $S^{2,0}\times S^{2,0}$ (corresponding to the involution $z\mapsto
z$) will give the type IIA theory on either $S^{2,0}\times S^{1,1}$
(corresponding to the involution $z\mapsto\bar{z}$), or $S^{1,1}\times
S^{2,0}$ (corresponding to the involution $z\mapsto -\bar{z}$). This
accounts for all of the species $2$ antiholomorphic involutions (see
Table \ref{Table:class_invo}). 

As can be seen from Table \ref{Table:class_invo}, the involutions
$z\mapsto\pm\bar{z}$ only correspond to $S^{1,1}\times S^{2,0}$ if the
complex modulus is $\tau=i\tau_2$ with $\tau_2\geq 1$. This tells us
that we must have a rectangular torus. Our $2$-torus is also equipped
with a K\"ahler form $J\equiv\sqrt{G}dx\wedge dy$ and the NS-NS
$2$-form $B$-field $B$, which combine to give the K\"ahler modulus
$\rho=\int_{\bT^2}(B+iJ)$. $T$-duality exchanges $\tau$ and $\rho$.
Since $\tau$ is purely imaginary in the type IIA theory, $\rho$ must be
purely imaginary in the $T$-dual theory. Therefore, the type I theory
compactified on a $2$-torus cannot have any $B$-field (the only non-zero
possibility for a $B$-field is $B=\frac{1}{2}$, which gives the type I
theory without vector structure as described in \cite{Witten:1998-02},
and will be discussed later). As always, we only consider the case where the type IIA theory has zero $B$-field so
that the $T$-dual IIB theory is on a rectangular torus as
well.\footnote{Assuming the $B$-field is trivial in the type IIA
  theories does not affect our end results by Proposition \ref{prop:IIAnoB}.}  

After $T$-dualizing one of the two circles in the type I theory we can
$T$-dualize the other circle. This corresponds to $T$-dualizing the
copy of $S^{2,0}$ in the type IIA theory on $S^{1,1}\times S^{2,0}$ or
equivalently, simultaneously $T$-dualizing both circles in the
original type I theory. This gives the type IIB theory on
$S^{1,1}\times S^{1,1}$ which corresponds to the spacetime involution
$z\mapsto -z$. This can be easily seen by composing the involutions
that describe the $2$ individual $T$-dualities, $z\mapsto\bar{z}$ and
$z\mapsto-\bar{z}$. The type IIB theory on $S^{1,1}\times S^{1,1}$ has
$4$ $O7^+$-planes located at the $4$ fixed points of $z\mapsto-z$
which correspond to the $2$-torsion points of the elliptic curve:
$0,\frac{1}{2},\frac{\tau}{2},$ and $\frac{1}{2}+\frac{\tau}{2}$. This
chain of dualities can be neatly displayed as in Figure
\ref{Fig:species2}. At the Gepner point corresponding to $\tau=i$ in
the type IIA theory, there is a rotational symmetry under
multiplication by $i$, so the
involutions $z\mapsto\pm\bar{z}$ are equivalent. This collapses the
horizontal line in 
Figure \ref{Fig:species2}, corresponding to the fact that the torus is
square and there is no difference between the $2$ circles. 

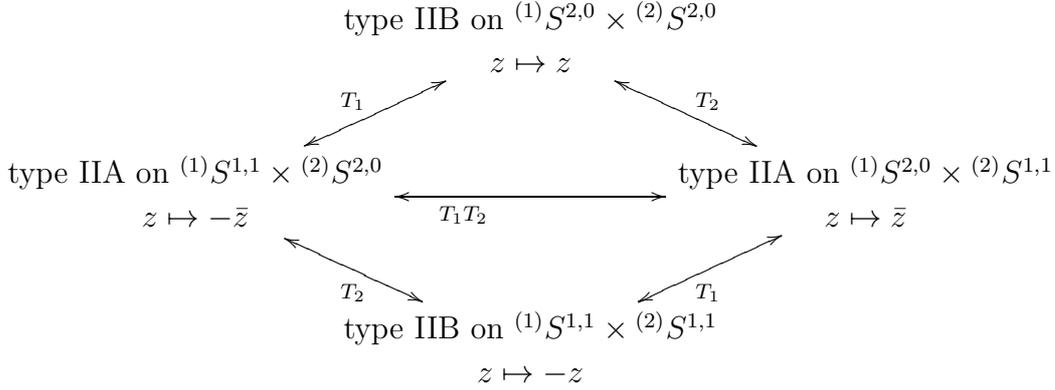
\begin{figure}[ht]
\centering
$\xymatrix@C-4pc{ & \genfrac{}{}{0pt}{0}{\text{type IIB on }{^{(1)}}S^{2,0}\times {^{(2)}}S^{2,0}}{z\mapsto z}\ar@{<->}[dl]_{T_{1}}\ar@{<->}[dr]^{T_{2}} &  \\
 \genfrac{}{}{0pt}{0}{\text{type IIA on }{^{(1)}}S^{1,1}\times {^{(2)}}S^{2,0}}{z\mapsto-\bar{z}}\ar@{<->}[dr]_{T_{2}}\ar@{<->}[rr]_(.4){T_1T_{2}}&  & \genfrac{}{}{0pt}{0}{\text{type IIA on }{^{(1)}}S^{2,0}\times {^{(2)}}S^{1,1}}{z\mapsto\bar{z}}\ar@{<->}[dl]^{T_{1}}\\ 
&  \genfrac{}{}{0pt}{0}{\text{type IIB on }{^{(1)}}S^{1,1}\times {^{(2)}}S^{1,1}}{z\mapsto-z} & }$
\caption{Chain of $T$-dualities connecting the various theories
  related to the type I theory with trivial $B$-field. $T_i$
  represents $T$-duality on the indicated circle.} 
\label{Fig:species2}
\end{figure}

\subsubsection{The type \~I theory}

There are a couple of ways we can compactify the type \~I theory on an
elliptic curve. The first way is to compactify the type \~I theory on
a single circle, and then on another circle with trivial
involution. For a single compact dimension, the type \~I theory is the
type IIB orientifold $(\bR^9\times S^1)/(\iota\cdot\Omega)$ where
$\iota$ is the spacetime 
involution that rotates the compact direction $\pi$ radians. In our
notation, this is the type IIB theory on $\bR^{9,0}\times S^{0,2}$. The
$T$-dual of the type \~I theory is the type $\widetilde{IA}$ theory
\cite[\S7.2]{Gao:2010ava}. The type IA theory contains $2$ 
$O8^+$-planes. The type $\widetilde{IA}$ theory is obtained from the
type IA theory by replacing one of the $O8^+$-planes with an
$O8^-$-plane. Using the notation $S^{p,q}_\alpha$ of \cite{Doran:2013sxa},
where $\alpha$ is the sign choice on the components
of the fixed set, the compactification manifold for the type \~I
theory is $S^{1,1}_{(+,-)}$. 

Let $x$ be the coordinate of the compact direction in the type \~I
theory. Considering the circle as $\bR/\bZ$, we see that $S^{0,2}$ is
the circle mod the involution 
$$x\mapsto x+\frac{1}{2}.$$
Under $T$-duality this becomes the dual circle mod the involution
$$\tilde x\mapsto -\tilde x+\frac{1}{2}.$$
The $2$ fixed points of this involution are located at
$x=\frac{1}{4},\frac{3}{4}$. We see that the $O$-planes are no longer
located at the $2$-torsion points $x=0$ and $x=\frac{1}{2}$, as they are with
the involution $x\mapsto -x$, but have been shifted. Every involution
$x\mapsto -x+\delta$, $\delta\in\bR$, gives $S^{1,1}$ with $2$
$O8^+$-planes except the case $\delta=\pm\frac{1}{2}$. What makes
$\delta=\frac{1}{2}$ unique is that the involution exchanges the $2$-torsion
points. For all other values of $\delta$, the two $2$-torsion points
are mapped to distinct points. The fact that the $2$-torsion points
are exchanged for $\delta=\frac{1}{2}$ corresponds physically to the
fact that the $O$-plane charges corresponding to the $2$-torsion
points can annihilate. However, the locations of the $O$-planes are
shifted, so we end up with an $O^+$-$O^-$-plane pair. This provides a
heuristic way of viewing the need for a twisting corresponding to a
sign choice, as will be discussed further in the following section. In
the language 
of \cite{Keurentjes:2000}, we should consider $S^{0,2}$ as a circle
with a crosscap attached and then we see that the $T$-dual of a
crosscap is an $O^+$-$O^-$ plane pair. 

Compactifying the type \~I theory on
another circle with trivial involution is the type IIB theory on
$\bR^{8,0}\times S^{0,2}\times S^{2,0}$. This corresponds to the
involution $z\mapsto z+\frac{1}{2}$ on $\bC/\Lambda$, 
with $S^{0,2}$ being the circle that is the image of $[0,1]$
and $S^{2,0}$ being the circle that is the image of $\tau\cdot[0,1]$. 

Now if we $T$-dualize the circle that is the image of $[0,1]$, $S^{0,2}$, we get
the type $\widetilde{IA}$ theory compactified on another circle with
trivial involution. In our notation this is the type IIA theory on
$\bR^{8,0}\times S_{(+,-)}^{1,1}\times S^{2,0}$ and corresponds to the
involution $z\mapsto -\bar{z}+\frac{1}{2}$. If we now $T$-dualize the
copy of $S^{2,0}$, we get the type IIB theory on
$S_{(+,-)}^{1,1}\times S_{(+,+)}^{1,1}=(S^{1,1}\times
S^{1,1})_{(+,+,-,-)}$. $(S^{1,1}\times S^{1,1})_{(+,+,-,-)}$ is 
$\bC/\Lambda$ with the involution $z\mapsto -z+\frac{1}{2}$. It has
$4$ fixed points corresponding to $2$ $O7^-$-planes and $2$
$O7^+$-planes. Let us now consider what happens if we perform the
$T$-dualities in the opposite order. 

If we first $T$-dualize the copy of $S^{2,0}$ in the type IIB
theory on $S^{0,2}\times S^{2,0}$, we get the type IIA on
$S^{0,2}\times S^{1,1}_{(+,+)}$ corresponding to the species $0$
antiholomorphic involution $z\mapsto\bar{z}+\frac{1}{2}$. If we now
$T$-dualize the copy of $S^{0,2}$ we will get the type IIB theory on
$S_{(+,-)}^{1,1}\times S_{(+,+)}^{1,1}$. This chain of dualities is shown in
Figure \ref{Fig:species0r}. 

\begin{figure}[h]
\centering
$\xymatrix@C-6pc{ & \genfrac{}{}{0pt}{0}{\text{type IIB on }{^{(1)}}S^{0,2}\times {^{(2)}}S^{2,0}}{z\mapsto z+\frac{1}{2}}\ar@{<->}[dl]_{T_{(1)}}\ar@{<->}[dr]^{T_{(2)}} &  \\
 \genfrac{}{}{0pt}{0}{\text{type IIA on }{^{(1)}}S_{(+,-)}^{1,1}\times {^{(2)}}S^{2,0}}{z\mapsto-\bar{z}+\frac{1}{2}}\ar@{<->}[dr]_{T_{(2)}}&  & \genfrac{}{}{0pt}{0}{\text{type IIA on }{^{(1)}}S^{0,2}\times {^{(2)}}S^{1,1}_{(+,+)}}{z\mapsto\bar{z}+\frac{1}{2}}\ar@{<->}[dl]^{T_{(1)}}\\ 
&  \genfrac{}{}{0pt}{0}{\text{type IIB on }{^{(1)}}S_{(+,-)}^{1,1}\times {^{(2)}}S^{1,1}_{(+,+)}}{z\mapsto-z+\frac{1}{2}} & }$
\caption{Chain of $T$-dualities connecting the various theories
  related to the species $0$ antiholomorphic involution
  $z\mapsto\bar{z}+\frac{1}{2}$. $T_i$ represents $T$-duality on the
  indicated circle.} 
\label{Fig:species0r}
\end{figure}
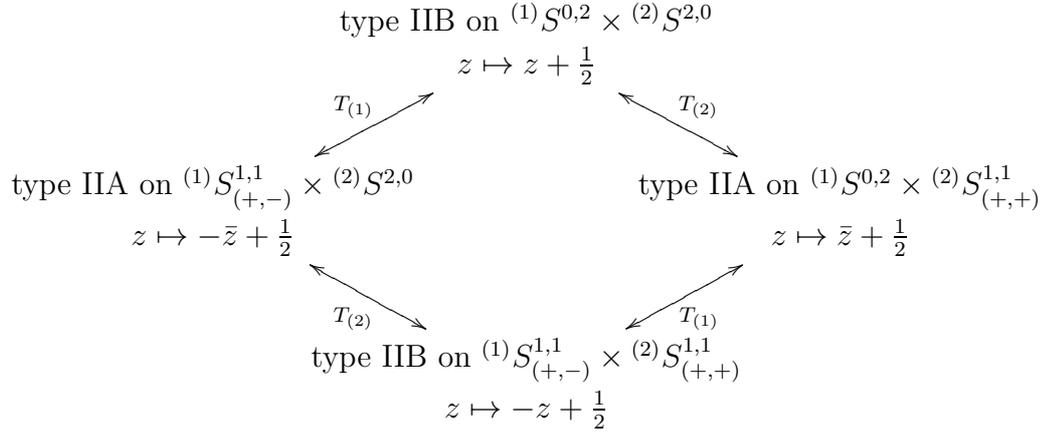

The group of theories related by $T$-duality pictured in Figure
\ref{Fig:species0r} doesn't contain all of the species $0$
antiholomorphic involutions and therefore does not have the symmetry
we saw with the group containing the type I theory with trivial $B$-field (Figure \ref{Fig:species2}). We can
easily obtain a picture containing the species $0$ antiholomorphic map
$z\mapsto-\bar{z}+\frac{\tau}{2}$ by just reversing the roles of
$\tau$ and $1$ in Figure \ref{Fig:species0r} by starting with the type
IIB theory with involution $z\mapsto z+\frac{\tau}{2}$. This, however,
requires multiple groupings and is not as satisfying a picture. 

This can be resolved by taking half shifts in both the real and
imaginary directions simultaneously. By this we mean starting with the
the type IIB theory with involution $z\mapsto
z+\frac{1+\tau}{2}$. This is the type IIB theory compactified on
$S^{0,2}\times S^{0,2}$. Performing a single $T$-duality in different
directions will give the type IIA theory on
$z\mapsto\pm\bar{z}+\frac{1+\tau}{2}$. We are being purposefully vague
about the direction of $T$-duality as it is not as simple as in the
previous cases and we will discuss it further shortly. This
corresponds to the type IIA theory on $S^{0,2}\times S_{(+,-)}^{1,1}$
or $S_{(+,-)}^{1,1}\times S^{0,2}$. 

The involution $z\mapsto\bar{z}+\frac{1+\tau}{2}$ is equivalent to
$z\mapsto\bar{z}+\frac{1}{2}$, while the involution
$z\mapsto-\bar{z}+\frac{1+\tau}{2}$ is equivalent to
$z\mapsto-\bar{z}+\frac{\tau}{2}$. Here, $2$ involutions, $\iota_1$
and $\iota_2$, are equivalent if $(\bC/\Lambda)/\iota_1$ and
$(\bC/\Lambda)/\iota_2$ are dianalytically equivalent as explained in
\cite{MR640091}. This means $\iota_1$ and $\iota_2$ are equivalent if
there exists an analytic automorphism, $\delta$, of $\bC/\Lambda$ such
that 
\begin{equation}
\iota_1=\delta\iota_2\delta^{-1}.
\end{equation}

Finally, $T$-dualizing the other direction gives us the type IIB
theory with involution $z\mapsto-z+\frac{1+\tau}{2}$ corresponding to
$S_{(+,-)}^{1,1}\times S_{(+,-)}^{1,1}=(S^{1,1}\times S^{1,1})_{(+,+,-,-)}$. This is again the type IIB
theory with $4$ fixed points corresponding to $2$ $O^-$-planes and $2$
$O^+$-planes. This chain of dualities is pictured in
Figure \ref{Fig:species0b}. 
                                                                                                                          
\begin{figure}[h]
\centering
$\xymatrix@C-6pc{ & \genfrac{}{}{0pt}{0}{\text{type IIB on }{^{(1)}}S^{0,2}\times {^{(2)}}S^{0,2}}{z\mapsto z+\frac{1+\tau}{2}}\ar@{<->}[dl]_{T_{1}}\ar@{<->}[dr]^{T_{2}} &  \\
 \genfrac{}{}{0pt}{0}{\text{type IIA on }{^{(1)}}S_{(+,-)}^{1,1}\times {^{(2)}}S^{0,2}}{z\mapsto-\bar{z}+\frac{1+\tau}{2}}\ar@{<->}[dr]_{T_{2}}\ar@{<->}[rr]_(.4){T_1T_{2}}&  & \genfrac{}{}{0pt}{0}{\text{type IIA on }{^{(1)}}S^{0,2}\times{^{(2)}}S_{(+,-)}^{1,1}}{z\mapsto\bar{z}+\frac{1+\tau}{2}}\ar@{<->}[dl]^{T_{1}}\\ 
&  \genfrac{}{}{0pt}{0}{\text{type IIB on }{^{(1)}}S^{1,1}_{(+,-)}\times{^{(2)}}S^{1,1}_{(+,-)}}{z\mapsto-z+\frac{1+\tau}{2}} & }$
\caption{Chain of $T$-dualities connecting the various theories related to the species $0$ antiholomorphic involution $z\mapsto\bar{z}+\frac{1+\tau}{2}$.}
\label{Fig:species0b}
\end{figure}
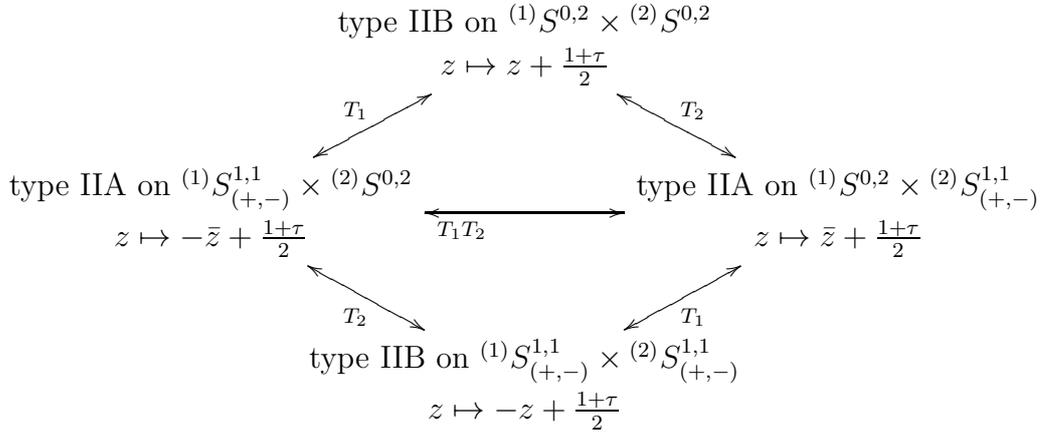

Note how, as opposed to Figure \ref{Fig:species0r}, Figure
\ref{Fig:species0b} is symmetric and both possible species $0$
antiholomorphic involutions occur in a single diagram. Also as in
the $T$-duality group pictured in Figure \ref{Fig:species2}, the
horizontal line in Figure \ref{Fig:species0b} 
collapses at the Gepner point $\tau=i$, signifying that the torus is
square. Since the species $0$ antiholomorphic involutions only exist
for $\tau$ purely imaginary, we see that the $T$-dual type IIB
theories must have trivial $B$-fields. As always, we are assuming that
the type IIA theories have trivial $B$-fields, so that the type IIB
tori are also rectangular. At first glance it might seem that the two
type IIB theories appearing at the top of Figures \ref{Fig:species0r}
and \ref{Fig:species0b} are related by a change of complex structure,
so that the type IIA theories appearing in Figure \ref{Fig:species0r}
differ by a choice of $B$-field from those appearing in Figure
\ref{Fig:species0b}. However, by Proposition \ref{prop:IIAnoB} we know
that the $B$-field does not affect the stable $D$-brane charge in the 
type IIA theories, so the $B$-field could not differentiate the IIA
theories appearing in Figures \ref{Fig:species0r} and
\ref{Fig:species0b}. We will see in Section \ref{sec:symnonsyms0} that
the theories in  Figures \ref{Fig:species0r} and \ref{Fig:species0b}
are related by the direction of $T$-duality. This could be viewed as a
change of complex structure, but as we will show the $2$ complex
structures are related by an $\SL(2,\bZ)$ transformation. Therefore
they are defined on equivalent elliptic curves and their $T$-dual
theories will have the same $B$-fields. Changing the direction of
$T$-duality will, however, alter the sources for the stable brane
charges. This effect was crucial in our understanding the sign choice
as a source of twisting. 

As an example of Proposition \ref{prop:IIAnoB} and to better explain
our assumption that all type 
IIA theories involved have trivial $B$-fields, let us look at the
$T$-duality group containing the type I theory a little more
closely. Proposition \ref{prop:IIAnoB} says that the assumption that
the type IIA theories have trivial $B$-field does 
not affect our results about $O$-plane and $D$-brane charges and their
relationship under $T$-duality. If we included a non-trivial $B$-field
in the type IIA theory on $S^{1,1}\times S^{2,0}$ for example, then it
will still be $T$-dual to the type I theory with trivial $B$-field
compactified on a $2$-torus. The only difference is that the torus
will no longer be rectangular, but this does not affect the stable brane
charges, nor the $O$-plane content (which determines the overall
theory). While the brane content in the type I theory will be the 
same as for the case when the torus is rectangular, since the
direction of $T$-duality is no longer orthogonal to (or in the same
direction as) the direction the branes wrap, the sources in the
$T$-dual theory could be affected. A brane that wraps both compact
directions is usually obtained from a brane that wraps a single
compact direction via $T$-duality by $T$-dualizing the direction
orthogonal to the direction wrapped by the brane. For a
non-rectangular torus, $T$-dualizing one leg can send a brane that
wraps a single cycle to a brane that wraps a different cycle. While
the source might change, the important feature is that the stable
charge remains the same. This will become clearer in Section
\ref{sec:branes} when we discuss the brane content in the various
theories. 

\subsection{$T$-duality for elliptic curve orientifolds with non-trivial $B$-field}

The only possible nonzero value for the $B$-field is
$B=\frac{1}{2}$. Furthermore, this is only a possibility for the type I
theory. The type \~I theory cannot have a non-trivial $B$-field. If it
did, its $T$-dual would be a non-rectangular torus, and as we just saw,
the type \~I theory is $T$-dual to the species $0$ real elliptic curve,
which only exists for rectangular tori. So we are left only to
consider the case of the type I theory with non-trivial $B$-field, or
as it is more commonly referred to, the type I theory without vector
structure \cite{Witten:1998-02}. 

Consider the type I theory without vector structure on a rectangular
torus. Its $T$-dual theory will be a type IIA theory on an elliptic
curve with $\Re{\tau}=\frac{1}{2}$ and an antiholomorphic
involution, and trivial $B$-field.\footnote{As usual the assumption that the type IIA theories have trivial $B$-field does not affect our final result by Proposition \ref{prop:IIAnoB}.} The
involution on the $T$-dual type IIA theory must be a species $1$
anti-holomorphic involution because those are the only possible
antiholomorphic involutions when $\Re{\tau}=\frac{1}{2}$. 

The torus with the species $1$ involution is the only torus orientifold
that cannot be split into the product of $2$ invariant circles 
and is the only truly new case we get from
considering compactifications on $2$ circles versus $1$. As can be
seen from Table \ref{Table:class_invo}, the species $1$
antiholomorphic involutions only exist for $\tau=e^{i\theta}$,
$\frac{\pi}{3}\leq\theta\leq\frac{\pi}{2}$, with involution
$z\mapsto\pm\tau\bar{z}$, or $\tau=\frac{1}{2}+i\tau_2$ with
$\tau_2\geq\frac{\sqrt{3}}{2}$ and involution $z\mapsto\pm\bar{z}$.
As noted previously, every real elliptic curve with $\tau=e^{i\theta}$
is equivalent to a real elliptic curve with
$\tau=\frac{1}{2}+i\tau_2$, and since this is the case we get from the
obvious way of $T$-dualizing the type I theory without vector
structure, it is the case we will consider. We will briefly consider
the case of $\tau=e^{i\theta}$ in the following section to motivate
our understanding of the sign choice as a twisting. 

So far we have seen that the type I theory without vector structure is
$T$-dual to the type IIA 
theory with $\tau=\frac{1}{2}+i\tau_2$, $\rho=iV$ and a species $1$
antiholomorphic involution. Performing another $T$-duality in the
other compact direction gives the type IIB theory with $4$ fixed 
points, but now with $3$ $O^+$-planes and $1$ $O^-$-plane. 

Without having our geometric description of the sign choice, it is
easiest to see this by noting that the species $1$ antiholomorphic
involutions give a M\"obius strip which can viewed as a cylinder with
a cross-cap. $T$-duality transforms the cross-cap into an $O^+$-,
$O^-$-plane pair, while the boundary is transformed to $2$
$O^+$-planes \cite{Keurentjes:2000}. We can put this chain of dualities into a 
diagram similar to the one for the species $0$ and $2$ groups, and the
result is given in Figure \ref{Fig:species1v}. 

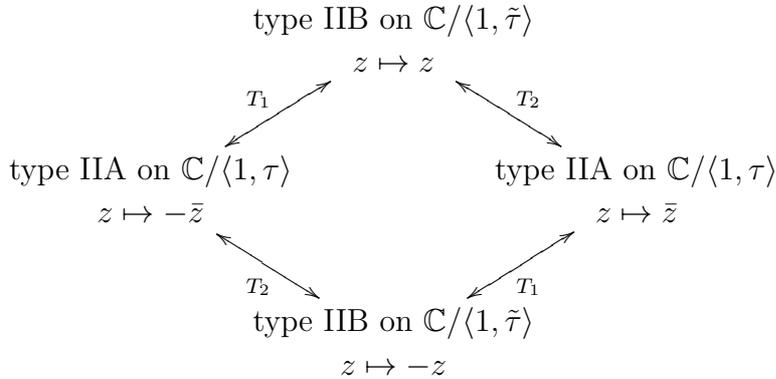
\begin{figure}[h]
\centering
$\xymatrix@C-4pc{ & \genfrac{}{}{0pt}{0}{\text{type IIB on }
\bC/\langle 1,\tilde{\tau}\rangle}{z\mapsto z}\ar@{<->}[dl]_{T_{1}}\ar@{<->}[dr]^{T_{2}} &  \\
 \genfrac{}{}{0pt}{0}{\text{type IIA on }\bC/\langle
   1,\tau\rangle}{z\mapsto-\bar{z}}\ar@{<->}[dr]_{T_{2}}&
 & \genfrac{}{}{0pt}{0}{\text{type IIA on
   }\bC/\langle 1,\tau\rangle }{z\mapsto\bar{z}}\ar@{<->}[dl]^{T_{1}}\\  
&  \genfrac{}{}{0pt}{0}{\text{type IIB on }\bC/\langle
   1,\tilde{\tau}\rangle}{z\mapsto-z} & }$ 
\caption{Chain of $T$-dualities connecting the various theories
  related to the type I theory without vector
  structure. $\tau=\frac{1}{2}+i\tau_2$ and $\tilde{\tau}$ is purely
  imaginary.} 
\label{Fig:species1v}
\end{figure}

Note that the involutions appearing in Figure \ref{Fig:species1v} are
the same as those appearing in the chain of dualities connecting the
group containing the type I theory (Figure \ref{Fig:species2}). This
shows that the two vertical legs in the left-hand side of Figure
\ref{Fig:modulispace} 
both describe the $T$-dualities of the type I theory, with the
$\Re{\tau}=0$ leg corresponding to trivial $B$-field, and the
$\Re{\tau}=\frac{1}{2}$ branch corresponding to $B=\frac{1}{2}$. In
\cite{Bates:2006} (and actually already in \cite{MR1168626})
the authors note that the species obstructs
continuous deformation from the large limit branch $\tau=i\infty$ to
the large limit branch $\tau=\frac{1}{2}+i\infty$, leading to two
disjoint large volume type IIB torus orientifolds. We see that these
two large volume type IIB theories correspond to the type I theory
with the two possible values for the $B$-field. While this shows a
physical relationship between the species $2$ and species $1$ groups,
the mathematical relationship is described by Alling in
\cite{MR640091}. This chain of dualities is considered in
\cite{Keurentjes:2000}.  

We still have not considered the arc associated to species $1$
appearing in Figure \ref{Fig:modulispace}. This corresponds to the
type IIA theory with $\tau=e^{i\theta}$,
$\frac{\pi}{3}<\theta<\frac{\pi}{2}$, and involution
$z\mapsto\pm\tau\bar{z}$. We have noted that these elliptic curves are
equivalent to ones with $\tau=\frac{1}{2}+i\tau_2$. Moreover, under
continuous transformations from the species $1$ antiholomorphic
involutions with $\tau=e^{i\theta}$ to the ones with $\tau=\frac{1}{2}+i\tau_2$
through the Gepner point, the $j$-invariant goes from positive to
negative through $0$ at the Gepner point $\tau=e^\frac{i\pi}{3}$. This
shows us that the 
brane content and dual theories stay the same except for values of the
parameters.  While considering the
type IIA theories with $\tau=\frac{1}{2}+i\tau_2$ is often more
convenient, the case with $\tau=e^{i\theta}$ will be useful later so we
show the chain of $T$-dualities related to this case in Figure
\ref{Fig:species1a}.  
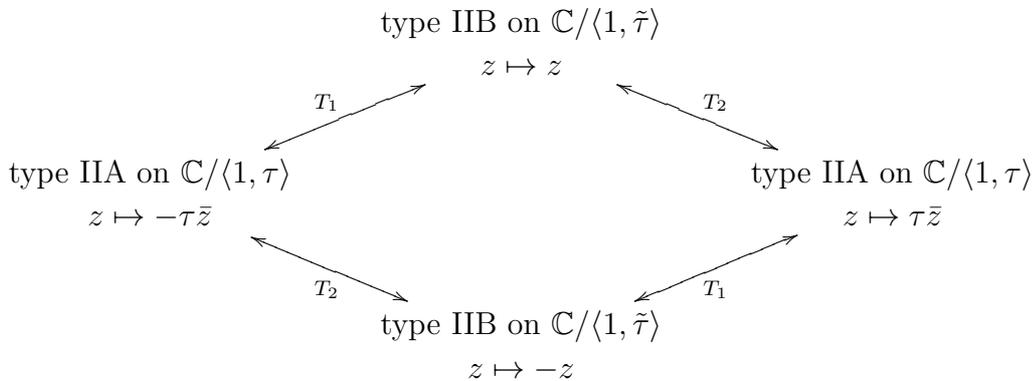
\begin{figure}[h!]
\centering
$\xymatrix{ & \genfrac{}{}{0pt}{0}{\text{type IIB on }\bC/\langle
      1,\tilde{\tau}\rangle}{z\mapsto
      z}\ar@{<->}[dl]_{T_{1}}\ar@{<->}[dr]^{T_{2}}
    &  \\ 
 \genfrac{}{}{0pt}{0}{\text{type IIA on
   }\bC/\langle 1,\tau\rangle}{z\mapsto-\tau\bar{z}}\ar@{<->}[dr]_{T_{2}}&
 & \genfrac{}{}{0pt}{0}{\text{type IIA on
   }\bC/\langle 1,\tau\rangle }{z\mapsto\tau\bar{z}}\ar@{<->}[dl]^{T_{1}}\\  
&  \genfrac{}{}{0pt}{0}{\text{type IIB on }\bC/\langle
   1,\tilde{\tau}\rangle}{z\mapsto-z} & }$ 
\caption{Chain of $T$-dualities connecting the various theories related to the species $1$ antiholomorphic involutions with $\tau=e^{i\theta}$. $\tilde{\tau}$ is purely imaginary.}
\label{Fig:species1a}
\end{figure}

\section{Jacobi functions, real elliptic curves, and $T$-duality}
\label{sec:geoTduality}

The different $T$-duality groupings described in the previous section all
contain spaces that are topologically equivalent, but have distinct
brane content. This is what motivated us to define $KR$-theory with a
sign choice in \cite{Doran:2013sxa}. While there are physical
explanations for the sign choices of the $O$-planes, so far we have seen no
mathematical explanation for why $T$-duality forces this type of
twisting on $KR$-theory. 

The standard twisting of $K$-theory by the $B$-field is usually
interpreted geometrically in terms of the angle the direction being
$T$-dualized makes with the cycles wrapped by branes. In this section
we will briefly mention how the distinction between the chains of
dualities pictured in Figures \ref{Fig:species0r} and
\ref{Fig:species0b}, as well as the $2$ branches of species $1$
antiholomorphic involutions, hints at such an interpretation. However,
we will not delve too far into this interpretation and instead give a
more general description of sign choices by describing canonical
normal forms for the elliptic curves appearing in the $3$ different
$T$-duality groups encoding this information. 

\subsection{A heuristic description of a sign choice as a twisting}
\label{sec:2tordes}
\subsubsection{Different ways of $T$-dualizing the type \~I theory}
\label{sec:symnonsyms0}
Let us first consider the distinction between the two different
$T$-duality groups containing the type \~I theory we discussed in
Section \ref{sec:Tduality} (Figures \ref{Fig:species0r} and
\ref{Fig:species0b}). When we discussed the chain of $T$-dualities
pictured in Figures \ref{Fig:species0r} and \ref{Fig:species0b} we
were purposefully vague about the directions of $T$-duality. For the
type IIB theory with involution $z\mapsto z+\frac{1}{2}$ (type 
IIB on $S^{0,2}\times S^{2,0}$) the circle given by the projection of $[0,1]$ is
equivariant and corresponds to the factor of $S^{0,2}$; the circle corresponding to the projection of $\tau\cdot[0,1]$, however, is not equivariant. The circle $\tau\cdot[0,1]$ is sent by the involution to the parallel circle
$\tau\cdot[0,1]+\frac{1}{2}$. By $T$-dualizing $S^{0,2}$ we get the type IIA theory on
$S_{(+,-)}^{1,1}\times S^{2,0}$. Now, the circle $[0,1]$
corresponds to $S_{(+,-)}^{1,1}$ and is still equivariant. Also, the
perpendicular circles, $\tau\cdot[0,1]+\frac{1}{4}$ and $\tau\cdot[0,1]+\frac{3}{4}$ are
also equivariant, making it clear what it means to $T$-dualize in the
$\tau$ direction. 

For the type IIB theory on $S^{0,2}\times S^{0,2}$ the equivariant circles
are the diagonal and anti-diagonal. In the $T$-dual theory
$S^{0,2}\times S_{(+,-)}^{1,1}$, the involution $z\mapsto\bar{z}+\frac{1+\tau}{2}$ exchanges the diagonal and
anti-diagonal. This shows that branes that wrap the diagonal and
anti-diagonal are not independent, and more importantly, branes
wrapping the real and imaginary axis are not 
independent. It is still well defined to $T$-dualize in the real and
imaginary directions, since the circles that go through the fixed
points in the type IIB theory with involution $z\mapsto
-z+\frac{1+\tau}{2}$ are all equivariant. This, combined with the
dependence between wrappings of the real and imaginary directions,
shows that branes that wrap $S^{0,2}$ in the non-symmetric case should
now wrap the diagonal (which is still a copy of $S^{0,2}$) and
continue to $T$-dualize in the real and imaginary directions. 

From the above discussion we are motivated to look at the equivalent
elliptic curve with $\tilde{\tau}=1+\tau$. $S^{0,2}\times S^{0,2}$
corresponds to a torus with legs $1$ and $it$, and involution
$z\mapsto z+\frac{1+it}{2}$. We could instead consider the equivalent
real elliptic curve given by a torus with legs $1$ and $1+it$, and the
same involution $z\mapsto z+\frac{1+it}{2}$. As pictured in Figure
\ref{Fig:species0twist} this decomposes $S^{0,2}\times S^{0,2}$ into
$S^{2,0}\times S^{0,2}$, where the copy of $S^{2,0}$ is generated by
$1$ and the copy of $S^{0,2}$ is generated by $1+it$. It might at
first seem unsatisfactory that the copy of $S^{2,0}$ is not
equivariant, but recall that for the involution $z\mapsto
z+\frac{1}{2}$ the generator of $S^{2,0}$, $\tau\cdot[0,1]$, is also not
equivariant but gets equated with the circle $\tau\cdot[0,1]+\frac{1}{2}$. In
the current situation, as pictured on the right-hand side of Figure
\ref{Fig:species0twist}, the circle $[0,1]$ is equated with the circle
$[0,1]+\frac{\tilde{\tau}}{2}=[0,1]+\frac{1+it}{2}$. 

\begin{figure}[h]
\centering
\begin{tabular}{c c c}
\begin{tikzpicture}[baseline]
\draw[black] (-1,-1) -- (1,-1);
\draw[black] (-1,1) -- (1,1);
\draw[blue] (-1,-1) -- (-1,1) [xshift=2pt] (-1,-1) -- (-1,1);
\draw[blue] (1,-1) -- (1,1) [xshift=2pt] (1,-1) -- (1,1);
\draw[red,dashed,very thick] (-1,-1) -- (1,1);
\end{tikzpicture}& $\;\;\simeq$ &
\begin{tikzpicture}[baseline]
\draw[green,decorate,decoration={snake,amplitude=.4mm,segment length=2mm}] (-1,-1) -- (1,-1);
\draw[black] (1,1) -- (3,1);
\draw[red,very thick] (1,-1) -- (3,1);
\draw[red,very thick] (-1,-1) -- (1,1);
\draw[blue,dashed] (1,-1) -- (1,1) [xshift=2pt] (1,-1) -- (1,1);
\draw[green,decorate,decoration={snake,amplitude=.4mm,segment length=2mm}] (0,0) -- (2,0);
\end{tikzpicture}
\end{tabular}
\caption{$S^{0,2}\times S^{0,2}$, corresponding to the involution $z\mapsto z+\frac{1+\tau}{2}$, on the left can be viewed as $S^{2,0}\times S^{0,2}$ on the right where the $2$ green squiggly lines are copies of $S^{2,0}$ that get exchanged under $z\mapsto z+\frac{1+\tau}{2}$.}
\label{Fig:species0twist}
\end{figure}
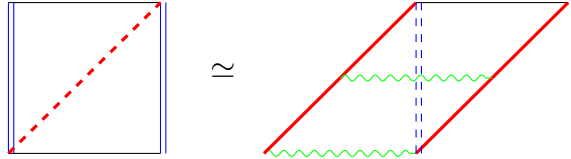

The non-symmetric species $0$ case can be obtained from the symmetric case by $T$-dualizing in a direction that isn't normal to the cycle the branes wrap. We see that the additional source of twisting beyond the
traditional $B$-field, the sign choice, is also related to the
direction of $T$-duality. The 
$B$-field is related to the angle between legs in the $T$-dual theory. A
trivial $B$-field corresponds to a rectangular torus in the $T$-dual
theory. In our present case all of the tori involved are rectangular (up to equivalence)
and thus all of the $B$-fields are trivial. From Figure
\ref{Fig:species0twist} it might appear as though the IIA theories
appearing in Figure \ref{Fig:species0b} will have non-trivial
$B$-field. However, the $2$ elliptic curves in Figure
\ref{Fig:species0twist} are equivalent. In fact they are defined by
the same lattice, with the only difference being the choice of
generators for the lattice. While the distinction between the chain of
dualities in Figures \ref{Fig:species0r} and \ref{Fig:species0b} can't
be explained by the $B$-field, we do expect a 
$B$-field-like twisting when we $T$-dualize a theory
containing branes that wrap a cycle that makes a non-right angle with
the direction of $T$-duality. This twisting is clearly accounted for
by the sign choice. 

Furthermore, note that there are $3$ different $T$-duality groupings
involving the type \~I theory on a rectangular torus. They correspond
to the $2$ asymmetric cases $z\mapsto z+\frac{1}{2}$ and $z\mapsto
z+\frac{\tau}{2}$, and the symmetric case $z\mapsto
z+\frac{1+\tau}{2}$. Both asymmetric groupings contain one species $0$
type IIA theory and one species $2$ type IIA theory twisted by a
non-trivial sign choice. The species $0$ type IIA theory is equivalent
to one of the $2$ species $0$ type IIA theories appearing in the
symmetric grouping. This shows us that the difference between the chain
of dualities pictured in Figure \ref{Fig:species0r} and the one
pictured in Figure \ref{Fig:species0b} is in the direction we
$T$-dualize. This choice of direction decides if we need to include a
twisting in one of the type IIA theories or if the twisting only
appears in the type IIB theory with $4$ fixed points. 

Perhaps more telling is the fact that, in all possible
compactifications of the type \~I theory on a $2$-torus, there was only
a single equivariant circle. After $T$-dualizing in the direction of
the equivariant circle, a new direction becomes equivariant, making it
well-defined to $T$-dualize in a second direction. But what if we want
to perform the $T$-dualities in the opposite order? Sticking with the
case pictured in Figure \ref{Fig:species0r} for concreteness, if we
want to $T$-dualize the $\tau$ direction first, then we must consider
pairs of branes that wrap the $2$ cycles $\tau\cdot[0,1]$ and $\tau\cdot[0,1]+\frac{1}{2}$
which will correspond to a single brane that wraps one of the new
equivariant cycles in $S^{1,1}_{(+,-)}\times S^{2,0}$ twice. 

$S_{(+,-)}^{1,1}\times S^{2,0}$ and
$S_{(+,-)}^{1,1}\times S^{0,2}$ are related via the annihilation of
$O^+$- and $O^-$-planes. Twisting by the sign choice is related to a
topological obstruction to the annihilation of the $O^+$- and
$O^-$-planes appearing in the type IIA theory on
$S_{(+,-)}^{1,1}\times S^{2,0}$. As with the case of the type \~I
theory compactified on a circle, we can determine the $O$-plane
content by looking at the $2$-torsion points. 

Under the involution
$z\mapsto z+\frac{1}{2}$ the $2$ torsion points transform as 
\begin{equation}
\xymatrix@R-1.3pc{0\ar@{<->}[r]& \frac{1}{2}\\
\frac{\tau}{2}\ar@{<->}[r]& \frac{1}{2}+\frac{\tau}{2}.}
\label{eq:2tortran}
\end{equation}
As was the case with single compact direction, the exchange of
$2$-torsion points corresponds to the fact that the $O$-planes
associated with them can annihilate. However, the locations of the
$O$-planes are shifted. The $T$-dual theory is the type IIA theory
with involution $z\mapsto-\bar{z}+\frac{1}{2}$. If the involution were
$z\mapsto-\bar{z}$ (the $T$-dual of $z\mapsto z$), then there would be
$2$ $O8^+$-planes wrapping the cycles $\tau$ (going through the
$2$-torsion points $0$ and $\frac{\tau}{2}$) and $\tau+\frac{1}{2}$
(going through the $2$-torsion points $\frac{1}{2}$ and
$\frac{\tau+1}{2}$). But for $z\mapsto z+\frac{1}{2}$, the $O$-planes
in the $T$-dual theory must have opposite charge since the $2$-torsion
points transform as in equation \eqref{eq:2tortran}. They do not
annihilate each other since they are shifted from the $2$ torsion
points and wrap the fixed circles $\tau\cdot[0,1]+\frac{1}{4}$ and
$\tau\cdot[0,1]+\frac{3}{4}$. 

For the type \~I theory, the action on the $2$-torsion points is
completely determined by the involution, which Atiyah's $KR$-theory is also
sensitive to. $T$-duality exchanges the complex and K\"ahler
moduli. Therefore the action of the original involution on the
$2$-torsion is no longer contained in the information of the new
$T$-dual involution, and $KR$-theory cannot pick this up. This is why
we need to add it in as an additional datum in the form of a
twisting. Before giving a more general description, let us do a similar
analysis for the different branches of the species $1$
antiholomorphic involutions. 

\subsubsection{Different ways of $T$-dualizing the type I theory
  without vector structure} 

In Section \ref{sec:Tduality} we saw that there were $2$ branches to
the species $1$ antiholomorphic involutions: one with
$\tau=e^{i\theta}$, $\frac{\pi}{3}\leq\theta\leq\frac{\pi}{2}$, and
the other with $\tau=\frac{1}{2}+i\tau_2$. We saw that every elliptic
curve with $\tau=e^{i\theta}$,
$\frac{\pi}{3}\leq\theta\leq\frac{\pi}{2}$, was equivalent to an
elliptic curve with $\tau=\frac{1}{2}+i\tau_2$. Furthermore, we saw
that after choosing this representative we could differentiate between
the $2$ branches by the sign of the $j$-invariant. While the
representative with $\tau=\frac{1}{2}+i\tau_2$ for all species $1$
real elliptic curves was useful for understanding the $B$-field in
$T$-dual theories, considering the differences in Figures
\ref{Fig:species1v} and \ref{Fig:species1a} is useful for
understanding the sign choice. 

As before we could determine the signs of the $O$-planes in the
$T$-dual theories by looking at the action of the antiholomorphic
involutions on the $2$-torsion points. Let us first consider the
$2$-torsion points in the type IIA theory with $\tau=e^{i\theta}$ with
involution $z\mapsto\tau\bar{z}$. Under this involution, 
the $2$-torsion points $0$ and $\frac{1}{2}+\frac{1}{2}e^{i\theta}$ 
are fixed while $\frac{1}{2}$ and $\frac{1}{2}e^{i\theta}$ are
exchanged. This shows us that in the $T$-dual type IIB  
theory with $4$ fixed points there is an $O^+$- $O^-$-plane pair
corresponding to the $2$-torsion points that were exchanged and 
there are two more $O^+$-planes corresponding to the two fixed $2$-torsion
points. There are $2$ possible independent $T$-dualities we could have
performed starting with the type IIA theory with
$\tau=e^{i\theta}$. The other $T$-duality would have taken us to a
space where the fixed set has a single component. Here the only
supersymmetric (physically significant) possibility is giving the
component the positive sign choice, corresponding to an orthogonal
structure on the Chan-Paton bundle. 

Let us now consider the $2$-torsion points under the action of the
species $1$ antiholomorphic involutions for
$\tau=\frac{1}{2}+i\tau_2$. Under the involutions $z\mapsto\pm\bar{z}$
the $2$-torsion points $0$ and $\frac{1}{2}$ are fixed, while
$\frac{\tau}{2}$ and $\frac{1}{2}+\frac{\tau}{2}$ are
exchanged, giving the same $O$-plane charge content as for the case of
$\tau=e^{i\theta}$. 

Here we can again view the need to include the extra twisting of a
sign choice in terms of the angle between the direction of $T$-duality
and the equivariant circles. The only 
equivariant circles for the species $1$ antiholomorphic involutions,
$z\mapsto\pm\tau\bar{z}$ are the diagonal, $S_D$, and anti-diagonal,
$S_A$. $S_D$ is the fixed circle for the involution
$z\mapsto\tau\bar{z}$ and $S_A$ is the fixed circle for the involution
$z\mapsto-\tau\bar{z}$. Therefore, the $O$-plane wraps either $S_D$ or
$S_A$ making an angle $\frac{\pi}{6}<\frac{\theta}{2}<\frac{\pi}{4}$
with the directions we are $T$-dualizing in, $1$ and $\tau$. For the
type IIA theory with $\tau=\frac{1}{2}+i\tau_2$ and involution
$z\mapsto\pm\bar{z}$ the equivariant circles are 
parallel to the real and imaginary axes, making non-orthogonal angles
with $\tau$.  

Note here that the fact that $\tau$ is not normal to the real axis is
the source of the $B$-field, but the angle between the equivariant
circle parallel to the imaginary axis and $\tau$ is the source for
this additional twisting of the sign choice. Now let us consider a
more general description that does not require analyzing each case
separately.

\subsection{Normal forms for real elliptic curves and a geometric
  description of the sign choice} 

As we just saw, we were able to determine the $O$-plane charges in a
couple of specific examples by looking at the action of the
antiholomorphic involutions on the $2$-torsion points of the elliptic
curve. We could follow a similar argument for each possible
case. However, this is not a good way to describe the sign choice. An
immediate disadvantage is that we would need to repeat the analysis
numerous times just to cover all of the cases in the $T$-duality grouping
containing the type \~I theory. Moreover, the description depends
on a choice of zero, which is unsatisfactory.\footnote{From the point
  of view of the physics, the group structure on the elliptic curve is
  not natural, since it requires fixing a distinguished basepoint in
  spacetime, violating Mach's principle, and we should therefore just keep the
  structure of a principal homogeneous space over the Jacobian.} 
We can use the
observation about the $2$-torsion points, however, to help us
determine a more general description. 

We will begin with the type IIA theories since these are defined on
real elliptic curves\footnote{Here 
  we take a more general definition of real elliptic curve, including both
  classical and non-classical elliptic curves in the sense of
  \cite[\S0.10]{MR640091}.} which have a canonical description in
terms of a defining equation 
for the field of meromorphic functions on the elliptic curve,
$E=\bR(x,y)$. The defining equations for elliptic curves are usually
given in terms of Weierstrass' elliptic function, $\wp$. This is not
the best way to describe real elliptic curves. While $E=\bR(\wp,\wp')$
for the species $2$ and $1$ real elliptic curves, this 
is not true for the species $0$ real elliptic curves since $\wp$ does
not have period $\frac{1}{2}$. Therefore $\wp\not\in E$
\cite{MR640091}. 

As noted by  Whittaker and Watson in \cite[Ch.\ XX]{MR1424469}, it is
easiest to use 
elliptic functions of order $2$ when proving general theorems about
elliptic functions, due to the behavior of their singularities. There
are two classes of order $2$ elliptic functions. The first class
contains order $2$ elliptic functions with a single double pole in
each fundamental domain. $\wp$ is in this class, and the fact that it
has a single 
pole is what prevents it from working in the species $0$ case. The
second class contains functions with $2$ simple poles whose residues
sum to zero in each cell. Clearly we would like to give the defining
equation for the species $0$ real elliptic curves in terms of elliptic
functions in class $2$, so that the shift of $\frac{1}{2}$ can
exchange the $2$ poles (accounting for the exchange of the $2$-torsion
points). 

As explained in \cite{MR640091}, the defining equation of the species
$0$ real elliptic curves can be written in terms of the standard
Jacobian elliptic function \emph{sinam}
(\emph{sinus amplitudinis}), denoted $\sn$. Note that if we let
the quarter period of $\sn$ be $K=\frac{1}{4}$ and the half-period of
$\sn$ be $K'=\frac{ti}{2}$ so that $\tau'=\frac{K'}{K}=2ti$, then we
can write $\sn$ in terms of theta functions as 
$$\sn{(u)}=\frac{\theta_0(0)\theta_1(2u)}{2\theta_1'(0)\theta_0(2u)}.$$
This makes sense since then $\sn$ has zeros at the points of $\Lambda$
and $\Lambda+\frac{1}{2}$, and poles at the points of
$\Lambda+\frac{ti}{2}$ and $\Lambda+\frac{1+it}{2}$. Note that $\sn$
has the same periods, $(1,ti)$, as the elliptic curve with $\tau=ti$
that the species $0$ involution is defined on. 

An immediate benefit of using the Jacobian elliptic functions instead
of the theta functions is that the theta functions are only periodic,
while the Jacobian elliptic functions are doubly periodic. Now
following \cite{MR640091}, we describe the defining equations for the
species $0$ real elliptic curves. 

Consider the real degree $4$ polynomial
\begin{equation}
L_{u,v,w}(x,k)\equiv (-1)^u(1-(-1)^vx^2)(1-(-1)^wk^2x^2),
\end{equation}
with $u,v,w\in\{0,1\}$, $0<k\leq 1$, and if $v=w$ then
$k<1$. Equations of the form 
$$y^2=L_{u,v,w}(x,k)$$
are said to be in generalized Legendre form.

The defining equation for the species $0$ real elliptic curves
can be put generalized Legendre form as 
\begin{equation}
\label{eq:deqs0}
(i\sn'(z))^2=L_{1,1,1}(i\sn(z),k),
\end{equation}
where as usual, $k$ is the Legendre modulus
$$k=\frac{\theta_2^2}{\theta_3^2},$$
where $\theta_i=\theta_i(0)$. Letting $x=i\sn(z)$ and $y=i\sn'(z)$,
\eqref{eq:deqs0} is
\begin{equation}
\label{eq:s0dq}
y^2=-(1+x^2)(1+k^2x^2).
\end{equation}

Before discussing how this relates to $T$-duality and sign choices, we
note that the species $2$ and species $1$ defining equations can also be put
in generalized Legendre form. Even though we can give a cubic defining
equation in terms of $\wp$ for the species $2$ and species $1$ real elliptic
curves, it will be more useful to use the quartic Legendre form both
for comparison to the species $0$ case and for general $T$-duality
analysis as well. 

The defining equation for the species $2$ antiholomorphic involutions
can be put in the generalized Legendre form 
\begin{equation}
\label{eq:deqs2}
(\sn'(z))^2=L_{0,0,0}(\sn(z),k).
\end{equation}
Letting $x=\sn(z)$ and $y=\sn'(z)$ this is
\begin{equation}
\label{eq:s2dq}
y^2=(1-x^2)(1-k^2x^2).
\end{equation}

The species $1$ real elliptic curves do not have purely imaginary
$\tau$, so we need to define $\sn$ in terms of different quarter- and
half-periods. For species $1$, let the quarter-period of $\sn$ in one
direction be $\frac{1}{4}$ while the half-period in the other
direction is $K'=\frac{\tau}{2}=\frac{1}{4}+\frac{i\tau_2}{2}$, so
$\tau'=\frac{K'}{K}=1+2i\tau_2$. This way $\sn$ has the same periods as
the elliptic curve the species $1$ involution is defined on,
$(1,\frac{1}{2}+i\tau_2)$. In this case the Legendre modulus, $k$, is
purely imaginary as described in \cite{MR640091}. We can put the
defining equation for the species $1$ real elliptic curves in the
generalized Legendre form as 
\begin{equation}
\label{eq:deqs1}
(\sn'(z))^2=L_{0,0,1}(\sn(z),-ik).
\end{equation}
Letting $x=\sn$ and $y=\sn'$ this is
\begin{equation}
\label{eq:s1dq}
y^2=(1-x^2)(1-k^2x^2).
\end{equation}
It might at first appear that equations \eqref{eq:s2dq} and
\eqref{eq:s1dq} are the same, but for equation \eqref{eq:s2dq},
$k^2>0$, while for equation \eqref{eq:s1dq}, $k^2<0$. 

Both the species $0$ and species $2$ antiholomorphic involutions only
exist on elliptic curves with purely imaginary complex
structure. Therefore, equations \eqref{eq:s0dq} and \eqref{eq:s2dq}
describe rectangular tori (as is clear by the choice of the period)
and it should be possible to perform both species $0$ and species $2$
involutions on either one. This naturally leads to the question of why
equation \eqref{eq:s0dq} is associated with species $0$ while equation
\eqref{eq:s2dq} is associated with species $2$. The answer (which turns
out to be a crucial ingredient in understanding $T$-duality and sign
choice) is determined by the effect of the different involutions on
$x$ and $y$.  

As we will see, equation \eqref{eq:s2dq} is associated with species $2$
because it is the elliptic curve where the meromorphic functions
$x$ and $y$ on $E$ (the generators of the function field, which we
will simply call the meromorphic coordinates) transform
in the standard way $(x,y)\mapsto(\bar{x},\bar{y})$ under the standard
species $2$ involution $z\mapsto\bar{z}$, and not for any other
involutions, while equation \eqref{eq:s0dq} is associated with species
$0$ since it is the elliptic curve for which the meromorphic
coordinates of $E$
transform in the standard way $(x,y)\mapsto(\bar{x},\bar{y})$ for the
species $0$ involution $z\mapsto\bar{z}+\frac{1}{2}$, and not for
any other involutions. In this way, equations \eqref{eq:deqs0} and
\eqref{eq:deqs2} can be thought of as the canonical normal forms for the
species $0$ and species $2$ real elliptic curves, respectively. While
the normal form for the species $1$ involutions is distinguished by
the complex structure (seen here by the quarter periods), we could
also view it as the canonical transformation for the species $1$
involutions since only under the standard involution $z\mapsto\bar{z}$
do the meromorphic coordinates of $E$ transform in the standard way
$(x,y)\mapsto(\bar{x},\bar{y})$. 

Now that we have a canonical choice of a normal form for each species
of antiholomorphic involution, we can give a general geometric
description of the $T$-duality relationships we saw in Section
\ref{sec:Tduality}. We will immediately see each $T$-duality grouping
should be classified by the generalized Legendre form of the elliptic
curve the involutions are defined on. 

At the Gepner point $\tau=i$ there exist antiholomorphic involutions
for all $3$ species. This fact allows one to find some non-trivial
identities for the elliptic functions by performing involutions of a
certain species on the elliptic curve not canonically associated with
that species. 

\subsubsection{The $T$-duality group defined by $y^2=(1-x^2)(1-k^2x^2)$}

As we saw above, the elliptic curve defined by 
$$y^2=L_{0,0,0}(x,k),$$
with $x=\sn(z)$ and $y=\sn'(z)$ has purely imaginary complex
structure parameter $\tau$. Therefore equation \eqref{eq:deqs2} has to be
associated to 
either the $T$-duality group containing the type I theory with trivial
$B$-field or the one containing the type \~I theory. To determine
which, we need to find the involution that sends $(\sn(z),\sn'(z))$ to
$(\overline{\sn(z)},\overline{\sn'(z)})$. 

Using the addition formulas for $\sn(z)$ and $\sn'(z)=\cn(z)\dn(z)$, as well as
the fact that $\sn$ is odd while $\cn$ and $\dn$ are even, it can be
easily shown that under complex conjugation 
\begin{equation}
\label{eq:snconj}
\sn(\pm\bar{z})=\pm\overline{\sn(z)}\;\;\;\;\text{and}\;\;\;\;\sn'(\pm\bar{z})=\overline{\sn'(z)}.
\end{equation}
Having shown that $z\mapsto\bar{z}$ sends $(x,y)$ to
$(\bar{x},\bar{y})$ we see that the elliptic curve defined by equation
\eqref{eq:s0dq} should be associated to the $T$-duality grouping that
contains $z\mapsto\bar{z}$ on a rectangular torus, i.e., to the grouping
containing the type I theory (Figure \ref{Fig:species2}). 

We should view all of the theories appearing in Figure
\ref{Fig:species2} as being defined on the elliptic curve  
$$y^2=(1-x^2)(1-k^2x^2),$$
differing only in the involution defining their orientifold (real)
structure and the Legendre modulus, which depends on $\tau$. With this viewpoint we can describe the chain of dualities
pictured in Figure \ref{Fig:species2} in terms of the induced action
on $x$ and $y$ alone as pictured in Figure \ref{Fig:species2gen}. 
\begin{figure}[h]
\centering
$\xymatrix{ & {(x,y)\mapsto (x,y)}\ar@{<->}[dl]_{T_{1}}\ar@{<->}[dr]^{T_{2}} &  \\
 {(x,y)\mapsto(-\bar{x},\bar{y})}\ar@{<->}[dr]_{T_{2}}&  &
 {(x,y)\mapsto(\bar{x},\bar{y})}\ar@{<->}[dl]^{T_{1}}\\ 
&  {(x,y)\mapsto(-x,y)} & }$
\caption{Chain of $T$-dualities connecting the various theories
  related to the type I theory with trivial $B$-field (top of the
  diagram). All theories are defined on the elliptic curve
  $y^2=(1-x^2)(1-k^2x^2)$. Holomorphic involutions correspond to type
  IIB theories and antiholomorphic involutions correspond to type IIA
  theories.}  
\label{Fig:species2gen}
\end{figure}
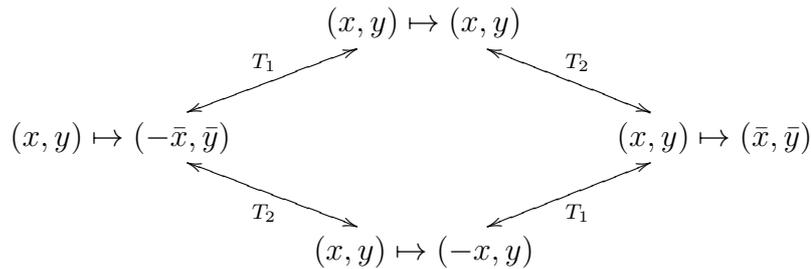

The only information we haven't yet specified is the sign choice. We
saw in Section \ref{sec:2tordes} that this information was obtained from
the action on the $2$-torsion points. This was, however, an unsatisfactory description, since as is clear from Figure \ref{Fig:species0twist} the description depended on distinguishing between equivalent elliptic curves. This same information is
contained in a \emph{canonical way} in the zeros of the
Jacobi Legendre normal form. 

The zeros of equation \eqref{eq:s2dq} are $x=\sn(z)=\pm 1$ and
$\pm\frac{1}{k}$. The fact that all of the zeros are real tells us
that all of the components of the fixed sets of any of the involutions
must have an orthogonal structure, or a positive sign choice. The fact
that the zeros are real means that they are fixed under conjugation,
corresponding to the fact that all of the $2$-torsion points are fixed in this
$T$-duality grouping, making the link between the action on the
$2$-torsion points and the type of zeros of $y$ explicit. 

Note that the zeros of $\sn'(z)$ occur at $z=\frac{1}{4}$,
$\frac{3}{4}$, $\frac{1}{4}+\frac{ti}{2}$ and
$\frac{3}{4}+\frac{ti}{2}$. Furthermore, $\sn$ is real and distinct at
all $4$ of those values. In particular, 
\begin{align}
\sn\left(\frac{1}{4}\right)&=1\\
\sn\left(\frac{3}{4}\right)&=-1\nonumber\\
\sn\left(\frac{1}{4}+\frac{ti}{2}\right)&=\frac{1}{k}\nonumber\\
\sn\left(\frac{1}{4}+\frac{ti}{2}\right)&=-\frac{1}{k}.\nonumber
\end{align}

Therefore we see that the $T$-duality grouping
containing the type I theory is associated to the normal form
$(\sn')^2=L_{0,0,0}(\sn,k)$ and the involutions pictured in Figure
\ref{Fig:species2gen}. We should note that all the theories are defined by the same normal form. $T$-duality changes the value of $\tau$, and hence the Legendre modulus. $T$-duality leaves the canonical normal form invariant, but changes the Legendre modulus and defining periods. Now let us turn our attention to the group
containing the type \~I theory. 

\subsubsection{The $T$-duality group defined by $y^2=-(1+x^2)(1+k^2x^2)$}

We have already seen that the elliptic curve defined by equation
\eqref{eq:deqs0} should be associated with a species $0$ involution, and
so must be associated with one of three possible versions of the
$T$-duality grouping that contains the type \~I theory. Based on the choice of $x$ and $y$, it should be associated with the type \~I theory
defined by a half-period shift in the real direction. Given our
definition of $\sn$ ($K=\frac{1}{4}$, $K'=\frac{it}{2}$), 
\begin{align}
\label{eq:realhalfshift}
\sn(z+\frac{1}{2})=\sn(z+2K)&=-\sn(z),\\
\sn'(z+2K)&=-\sn'(z).\nonumber 
\end{align}
Therefore under $z\mapsto\bar{z}+\frac{1}{2}$
\begin{align}
(x,y)=(i\sn(z),i\sn'(z))\mapsto &\left(i\sn\left(\bar{z}+\frac{1}{2}\right),
i\sn'\left(\bar{z}+\frac{1}{2}\right)\right)\\
&=(-i\sn(\bar{z}),-i\sn'(\bar{z}))\nonumber\\
&=(-i\,\overline{\sn(z)},-i\,\overline{\sn'(z)})\nonumber\\
&=(\bar{x},\bar{y}),
\end{align}
showing that the elliptic curve defined by equation \eqref{eq:deqs0}
should be associated with the $T$-duality grouping pictured in Figure
\ref{Fig:species0r}. Again, we can view this $T$-duality group in
terms of the induced action on $x$ and $y$ alone, as pictured in
Figure \ref{Fig:species0rNf}.
\begin{figure}[h]
\centering
$\xymatrix{ & {(x,y)\mapsto (-x,-y)}\ar@{<->}[dl]_{T_{1}}\ar@{<->}[dr]^{T_{2}} &  \\
 {(x,y)\mapsto(-\bar{x},\bar{y})}\ar@{<->}[dr]_{T_{2}}&  & {(x,y)\mapsto(\bar{x},\bar{y})}\ar@{<->}[dl]^{T_{1}}\\ 
&  {(x,y)\mapsto(x,-y)} & }$
\caption{Chain of $T$-dualities connecting the various theories
  related to the type \~I theory on $S^{0,2}\times S^{2,0}$ (top of
  the diagram). All theories are defined on the elliptic curve
  $y^2=-(1+x^2)(1+k^2x^2)$. Holomorphic involutions correspond to type
  IIB theories and antiholomorphic involutions correspond to type IIA
  theories.}  
\label{Fig:species0rNf}
\end{figure}

As before, we can determine the sign choice from the zeros of equation
\eqref{eq:deqs0}. The zeros are $x=\pm i$ and $\pm\frac{i}{k}$. Since
$x=i\sn(z)$ we see that the zeros occur at the same places as for
equation \eqref{eq:deqs2}: $z=\frac{1}{4}$, $\frac{3}{4}$,
$\frac{1}{4}+\frac{ti}{2}$ and $\frac{3}{4}+\frac{ti}{2}$. However,
now instead of the all of the zeros being real, they are imaginary and
come in complex conjugate pairs. This tells us that the charges of
the $O$-planes associated to these points must have opposite charge,
corresponding in our previous language to the fact that the
corresponding $2$-torsion points are exchanged. Let us first consider the
theory at the bottom of Figure \ref{Fig:species0rNf}, the type IIB
orientifold with with $4$ fixed points. The fixed points are at the
zeros of $i\sn'(z)$. The fact that $i\sn(\frac{1}{4})$ and
$i\sn(\frac{3}{4})$ are complex conjugates means that the $O$-planes
located there are opposite in sign. We can define our signs so that the sign
of the $O$-plane at each fixed point is the same as the sign of $\sn$
evaluated at that point. 

Now let us consider the antiholomorphic involutions, or type IIA
theories. The theory corresponding to $(x,y)\mapsto(\bar{x},\bar{y})$
is fixed point free since
$(i\sn(z),i\sn'(z))=(-i\overline{\sn(z)},-i\overline{\sn'(z)})$ has no
solutions. The type IIA theory with involution 
$(x,y)\mapsto(-\bar{x},\bar{y})$ in Figure \ref{Fig:species0rNf} has
$2$ $O$-planes. One wraps $\tau+\frac{1}{4}$ and the other wraps
$\tau+\frac{3}{4}$. Parametrize the 
cycle through $\tau+\frac{1}{4}$ by $s\in\left[0,1\right)$ as
$z=\frac{1}{4}+its$. Then it is easy to see $\sn$ is real on
$\tau+\frac{1}{4}$ and $\tau+\frac{3}{4}$ and has opposite sign on the
two cycles. Therefore $x=i\sn$ evaluated on the two cycles are
complex conjugates of each other, showing that the $2$ $O$-planes
should have opposite sign. This can be seen immediately by noting
one of the cycles goes through the two zeros of $y$ with positive
sign, while the other goes through the two zeros with a negative sign.  

Before considering the symmetric variant of the $T$-duality group
containing the type \~I theory (Figure \ref{Fig:species0b}), let's
consider a half-shift in the imaginary direction alone. That is,
consider the involution $z\mapsto z+\frac{\tau}{2}$ where $\tau=it$. Given our
definition of $\sn$, 
\begin{align}
\sn\left(z+\frac{it}{2}\right)=\sn(z+K')&=\frac{1}{k\sn(z)},\\
\sn'(z+K')&=\frac{-\sn'(z)}{k(\sn(z))^2}
\end{align}
Therefore, under $z\mapsto z+\frac{\tau}{2}$, $(x,y)\mapsto
(-\frac{1}{kx},\frac{y}{kx^2})$. Note that this in fact an
automorphism of the elliptic curve $y^2=-(1+x^2)(1+k^2x^2)$. In many
ways it is more interesting than the other automorphisms we have
encountered so far (conjugation, and multiplication by $-1$). However,
for our current purposes, it is unsatisfying that a shift in the imaginary
direction is distinguished. This is because we made made a choice of
preferred direction by choosing to use $\sn$ to define the elliptic
curve. 

As noted, $\sn$ has poles at $\Lambda+\frac{ti}{2}$ and
$\Lambda+\frac{1+it}{2}$. We chose $\sn$ because its $2$ poles where
exchanged by a shift by $\frac{1}{2}$. If we wanted to perform a
similar analysis for $z\mapsto z+\frac{\tau}{2}$ we would need to use
a Jacobi function with poles that are exchanged by a shift by
$\frac{\tau}{2}$. This leads us immediately to $\jsc$.

This is equivalent to exchanging the roles of the real and imaginary axes, since
\begin{equation}
i\jsc(z,k')=\sn(iz,k),
\end{equation}
where $(k')^2+k^2=1$. Letting $K'=\frac{\tau}{4}=\frac{it}{4}$ and
$K=\frac{1}{2}$ we see that the real elliptic curve with involution
$z\mapsto -\bar{z}+\frac{\tau}{2}$ is described by 
\begin{equation}
(\jsc'(z,k))^2=L_{1,1,1}(\jsc(z,k),k').
\end{equation}
Letting $y=\jsc'(z)$ and $x=\jsc(z)$ this becomes
$$y^2=(1+x^2)(1+k'^2x^2).$$
With this definition, we see that the $T$-duality group containing the
type \~I theory defined by a half-period shift in the imaginary
direction takes the same form as the $T$-duality group depicted in
Figure \ref{Fig:species0rNf}. Again the zeros take the same form but
with the roles of the real and imaginary axes reversed. That is the
zeros occur at $\frac{it}{4}$, $\frac{1}{2}+\frac{it}{4}$,
$\frac{3it}{4}$, and $\frac{1}{2}+\frac{3it}{4}$, and they come in two
pairs of complex conjugates. Note that we could also view the
canonical normal form for the type I $T$-duality group in terms of
$\jsc$ as  
$$(i\jsc')^2=L_{0,0,0}(i\jsc,k).$$
In this way we can view the difference between the species $0$ and
species $2$ antiholomorphic involutions as exchanging the roles of the
imaginary and real axes. 

Now returning to the symmetric case, we are tempted to use $\sd$ with
$K'=\frac{\tau}{4}$ and $K=\frac{1}{4}$, so that the poles are
exchanged by a shift along the diagonal. However this does not quite
work out, as will become clearer when we discuss the $T$-duality group
containing the type I theory without vector structure. The problem is
that $\sd$ satisfies the differential equation 
$$(\sd')^2=(1-k'^2\sd^2)(1+k^2\sd^2),$$
making it clear that there are $2$ real zeros and $2$ imaginary
zeros. This is not the correct form we expect for determining the sign
choice. This is rectified by instead using
$K'=\frac{1}{4}+\frac{it}{4}$ and $K=\frac{1}{4}$, so that
$\tau'=\frac{1}{2}+\frac{it}{2}$. Then following the same argument for
the species $1$ antiholomorphic involution from \cite{MR640091}, $k$ is
purely imaginary, while $k'=\sqrt{1-k^2}$ is real. Therefore the
zeros have the desired form if we let $x=i\sd(z)$ and
$y=i\sd'(z)$. While this can be done, we see that the 
asymmetric $T$-dual group containing the type \~I theory is more
natural. This is because we have to make a choice of direction no
matter which case we consider.

\subsubsection{The $T$-duality group defined by $y^2=(1-x^2)(1-k^2x^2)$, $k^2<0$}

As described above the elliptic curve $y^2=(1-x^2)(1-k^2x^2)$, with
$k^2<0$, $y=\sn'$ and $x=\sn$ is associated to the $T$-duality group
containing the type I theory without vector structure. As usual we can
describe the entire $T$-duality group (Figure \ref{Fig:species1v}) in
terms of the action on $x$ and $y$. Since $x$ and $y$ have the same
definitions as for the $T$-duality group containing the type I theory,
$x$ and $y$ will transform in the same way. However, the normal form
now has two real zeros and $2$ complex zeros which are complex
conjugates of each other, which determines the sign choice of
$(+,+,+,-)$ for the only non-trivial case appearing in this grouping. 
\begin{figure}[h]
\centering
$\xymatrix{ & {(x,y)\mapsto (x,y)}\ar@{<->}[dl]_{T_{1}}\ar@{<->}[dr]^{T_{2}} &  \\
 {(x,y)\mapsto(-\bar{x},\bar{y})}\ar@{<->}[dr]_{T_{2}}&  & {(x,y)\mapsto(\bar{x},\bar{y})}\ar@{<->}[dl]^{T_{1}}\\ 
&  {(x,y)\mapsto(-x,y)} & }$
\caption{Chain of $T$-dualities connecting the various theories
  related to the type I theory with non-trivial $B$-field (top of the
  diagram). All theories are defined on the elliptic curve
  $y^2=(1-x^2)(1-k^2x^2)$, $k^2<0$. Holomorphic involutions correspond to type
  IIB theories and antiholomorphic involutions correspond to type IIA
  theories.}  
\label{Fig:species1NF}
\end{figure}
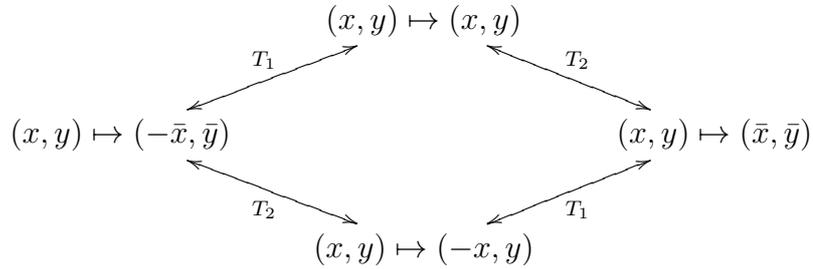
The distinction between the $T$-duality groups containing the type I
theory and type I theory without vector structure can be distinguished
by whether $k$ is real or imaginary. With this in mind and using the
identity 
$$\sn(z,ik)=k_1'\sd(z/k_1',k_1),$$
where $k_1=\frac{k}{\sqrt{1+k^2}}$ and $k_1k_1'=\frac{k}{1+k^2}$, we see that the elliptic curve can be written as 
$$y^2=(1-k_1'^2x^2)(1+k_1^2x^2),$$
making the form of the zeros clearer.

Now that we have given a geometric description for all of the possible
$T$-duality groups on an elliptic curve and the relevant sign choices,
we can describe the brane content in all of the various theories.

\section{D-brane content in the various orientifold theories}
\label{sec:branes}

Let's begin by reviewing all the twisted $KR$ groups for the
elliptic curve orientifolds. The results are given in Table
\ref{KRtable}, with the $T$-duality groupings color-coded.
We will want to analyze this table to determine the
D-brane content in each theory, and how the D-branes transform under
$T$-duality.   Note that this table only includes the
\emph{topological} type of each involution, and doesn't include
information on the complex structure. When the fixed set of the
involution is disconnected, the charges of the various $O$-planes are
indicated. 

\begin{table}[ht]
\begin{center}
\renewcommand{\arraystretch}{1.1}
\begin{tabular}{>{\columncolor{light-blue}}ll>{\columncolor{yellow}}l
>{\columncolor{cy}}l}
\rowcolor{green}
{\textbf{\small Type}} & \parbox{1in}{\strut\textbf{\small Fixed Set\\ and Twisting}} & 
{\textbf{Real Space}} & {\textbf{$KR$ Groups}}\\
{IIB (I)} & {$T^2$}&{$S^{2,0}\times
  S^{2,0}$}&\cellcolor{yel}{$KO^*(T^2)$}\\
{IIB (I no vec.)} & {$T^2$ with $w_2$}&{$S^{2,0}\times
  S^{2,0}$}&\cellcolor{light-red}{$KO^{*-1}\oplus KO^{*-1}\oplus K^*$}\\
{IIB (\~I)} & {$\emptyset$}&{$S^{2,0}\times
  S^{0,2}$}&{$KSC^*\oplus KSC^{*-1}$}\\ 
{IIB} & {$\{++++\}$}
  &{$S^{1,1}\times S^{1,1}$}&\cellcolor{yel}{$KO^{*+2}(T^2)$}\\
{IIB} & {$\{+++-\}$}&{$S^{1,1}\times 
  S^{1,1}$}&\cellcolor{light-red}{$KO^{*+1}\oplus KO^{*+1}\oplus K^{*}$}\\   
{IIB} & {$\{++--\}$}&{$S^{1,1}\times 
  S^{1,1}$}&{$KSC^{*+2}\oplus KSC^{*+1}$}\\
{IIA (species 2)} & {$S^1\amalg
  S^1$}&{$S^{1,1}\times S^{2,0}$}&\cellcolor{yel}{$KO^{*+1}(T^2)$}\\  
{IIA} & {$S^1_+\amalg
  S^1_-$}&{$S^{1,1}\times
  S^{2,0}$}&{$KSC^{*+1}\oplus KSC^{*}$}\\
{IIA (species 0)} &
{$\emptyset$}&{$S^{1,1}\times S^{0,2}$}
&{$KSC^{*+1}\oplus KSC^*$}\\     
{IIA (species 1)} &  {$S^1$}&{not a
  product}&\cellcolor{light-red}{$KO^*\oplus KO^*\oplus K^{*-1}$}\\ 
\end{tabular}
\end{center}
\caption{Summary of the twisted $KR$ groups for all the
elliptic curve orientifolds, with the T-duality groupings color-coded}
\label{KRtable}
\end{table}

The first thing one notes about the table is that the torsion-free
part of the $KR$-groups is the same in all cases, except for a degree
shift which is accounted for by T-duality. This is perhaps clearer in
Table \ref{Table:charges}, obtained from the $KR$ calculations via
equation \eqref{eqn:relvabs}. This table will be explained more fully below
when we discuss the specific brane content.  For example, the
torsion-free part of $KO^*$ is $\bZ$ in degrees $0$ mod $4$, so
the torsion-free part of $KO^*(T^2)$, which classifies the D-brane
charges in the type I theory with trivial B-field, is $\bZ$ in
degrees $0$ and $2$ mod $4$, and $\bZ^2$ in degrees $1$ mod $4$.
This is the same as the torsion-free part of the first column in Table
\ref{table:twistedKO}, which classifies the D-brane
charges in the type I theory with no vector structure, and also the
same as the torsion-free part of $KSC^*\oplus KSC^{*-1}$, which
classifies D-brane charges in the species $0$ case.

\begin{table}[ht]
\centering
\noindent\makebox[\textwidth]{\begin{tabular}{|| c || c | c | c ||}
\hline
$T$-duality group & Type I & Type I no vector & Type 
\~I\rule{0pt}{14pt}\\ \hline
{\small Real zeros in normal form}& $4$ & $2$ & $0$ \\ \hline\hline
{$Dp$-brane} & \multicolumn{3}{ c ||}{Charge} \\ 
\hline\hline
$D7$ & $\bZ\oplus\bZ_2^3$ & $\bZ\oplus\bZ_2^2$ & $\bZ\oplus\bZ_2$  \\ \hline
$D6$ & $\bZ_2^3$ & $\bZ_2^2$ & $\bZ_2$ \\ \hline
$D5$ & $\bZ\oplus\bZ_2$ & $\bZ$ & $\bZ$ \\ \hline
$D4$ & $\bZ^2$ & $\bZ^2$ & $\bZ^2$\\ \hline
$D3$ & $\bZ$ & $\bZ$ & $\bZ\oplus\bZ_2$\\ \hline
$D2$ & $0$ & $0$ & $\bZ_2$ \\ \hline
$D1$ & $\bZ$ & $\bZ$ & $\bZ$ \\ \hline
$D0$ & $\bZ^2\oplus\bZ_2$ & $\bZ^2$ & $\bZ^2$ \\ \hline
$D(-1)$ & $\bZ\oplus\bZ_2^3$ & $\bZ\oplus\bZ_2^2$ & $\bZ\oplus\bZ_2$  \\ \hline
\end{tabular}}
\smallskip
\caption{$D$-brane charges in all three groups.}
\label{Table:charges}
\end{table}

The explanation for this is that the \emph{torsion-free part} of the
twisted $KR$-groups classifies the BPS D-branes.  These are insensitive
to the $O$-plane charges, since the BPS planes are stable near both
$O^+$-planes and $O^-$-planes. Since each $T$-duality grouping
contains a IIB theory with four $O$-planes, which differ from each
other only in the $O$-plane charges, the BPS spectrum must be the same
in all cases.

Alternatively, one can argue that the BPS spectrum, being a
torsion-free phenomenon, does not depend on any twistings, either in
$H^2(T^2, \bZ_2)$ (this type of twisting distinguishes the type I
theory without vector structure from the usual type I theory) or
depending on a sign choice (since $KO$ and $KSp$ agree up to
torsion). We now look at the sources of the different brane charges in
all three $T$-duality groups. 

\subsection{The $T$-duality group defined on $y^2=(1-x^2)(1-k^2x^2)$}
\label{sec:Tduality2}

As described in the previous section, the $T$-duality group defined on
the elliptic curve $y^2=(1-x^2)(1-k^2x^2)$ contains the type I theory 
compactified on $\bT^2$ with trivial $B$-field. Before looking at this
case, it is useful to 
review the case of compactifying on a single circle. There are many
good sources for the $KR$-theory of single circle and its relation to
string theory \cite{Olsen:1999,Bergman:1999}.  

The type I theory compactified on a circle corresponds to the type IIB
theory compactified on $S^{2,0}$ and modded out by the action of
$\Omega$. As usual we will not explicitly state that we are modding out by $\Omega$ each time, since we will always be modding out by the action of $\Omega$. 

$Dp$-brane charges in the type I theory compactified on a circle are classified by
\begin{equation}
KR(S^{9-p,0}\times S^{2,0}, S^{2,0})\cong KO^{p-8}\oplus KO^{p-9},
\end{equation}
where $KO^{-j}=KO^{-j}(\pt)$. The second factor on the right-hand side
corresponds to $Dp$-brane charge coming from unwrapped branes and the
first factor corresponds to the charge contribution from wrapped
branes. The complete brane content is given in Table
\ref{Table:species2_S1}. 

Since the type IA theory is obtained from the type I theory
compactified on a circle by a $T$-duality, the relevant $KR$-theory is
shifted in index by $1$. Therefore, $Dp$-brane charges in the type IA
theory are classified by 
\begin{equation}
KR^{-1}(S^{9-p,0}\times S^{1,1}, S^{1,1})\cong KO^{p-9}\oplus KO^{p-8},
\end{equation}
where the second factor on the right-hand side corresponds to
$Dp$-brane charge coming from unwrapped branes and the first factor
corresponds to the charge contribution from wrapped branes. The
complete brane content is given in Table \ref{Table:species2_S1}. The
fact that $T$-duality exchanges wrapped and unwrapped branes is
described by the exchanged roles for $KO^{p-8}$ and $KO^{p-9}$ in the two
theories. 

\begin{table}[h]
\noindent\makebox[\textwidth]{\begin{tabular}{|| c || c | c | c | c | c | c | c | c | c | c || m{2.25cm} | m{2.25cm} ||}
\hline
$Dp$-brane & $D8$ & $D7$ & $D6$ & $D5$ & $D4$ & $D3$ & $D2$ & $D1$ & $D0$ & $D(-1)$ & type I on $S^1$ & type IIA on $S^{1,1}$ \\
\hline\hline
$KO^{p-8}$ & $\bZ$ & $\bZ_2$ & $\bZ_2$ & $0$ & $\bZ$ & $0$ & $0$ & $0$ & $\bZ$ & $\bZ_2$ & $(p+1)$-brane wrapping $S^{2,0}$  & unwrapped $p$-brane\\ \hline
$KO^{p-9}$ & $\bZ_2$ & $\bZ_2$ & $0$ & $\bZ$ & $0$ & $0$ & $0$ & $\bZ$ & $\bZ_2$ & $\bZ_2$ & unwrapped $p$-brane & $(p+1)$-brane wrapping $S^{1,1}$  \\ \hline
\end{tabular}}
\smallskip
\caption{$D$-brane charges in the type I theory compactified on a circle and the type IA theory.}
\label{Table:species2_S1}
\end{table}

$D0$-brane charge in the type I theory receives an integral
contribution from a wrapped BPS $D1$-brane and a $\bZ_2$ contribution
from an unwrapped non-BPS 
$D0$-brane. Under $T$-duality the wrapped $D1$-brane gets mapped to an
unwrapped $D0$-brane in the type IA theory and the $D0$-brane gets
mapped to a wrapped non-BPS $D1$-brane. However, the non-BPS branes
are not stable at all points of the moduli space, so we cannot extend
this argument to the entire moduli space. 

To see how this isomorphism is explained physically when the non-BPS
branes are unstable, let us look at the unwrapped $D0$-brane in the
type I theory, following \cite{Bergman:1999}. The spectrum of open
strings beginning and ending on the $D0$-brane is tachyon free in $10$
dimensions. However, when we compactify on a circle, the ground state
with winding number $1$ will have a classical mass squared given by 
$$m^2=-\frac{1}{2}+R^2,$$
in units with $\alpha'=1$. It will therefore be tachyonic if the
radius of the compactification circle is $R<\frac{1}{\sqrt{2}}$. In
this situation, the $D0$-brane will decay into a $D1$-$\overline{D1}$
pair that wrap the $S^1$. The tachyon must have anti-periodic boundary
conditions so that above the critical radius it will condense into a
stable kink (the $D0$-brane). This requires turning on a $\bZ_2$
Wilson line on either the $D1$-brane or $\overline{D1}$-brane. The
$\bZ_2$ charge of the unwrapped $D0$-brane corresponds to a $\bZ_2$
valued Wilson line in its decayed configuration. The same argument
shows that the $\bZ_2$ charge of the unwrapped $D7$- and $D8$-branes
correspond to $\bZ_2$ valued Wilson lines on their decay
configurations, wrapped $D8$-$D8$ and $D9$-$\overline{D9}$ pairs
respectively. For the $D(-1)$- and wrapped $D0$-brane you have to
compare the instanton action since they are instantonic. 

Under $T$-duality the unwrapped non-BPS $Dp$-branes with $p=0,7,8$ get
mapped to wrapped $D(p+1)$-branes. Since $T$-duality inverts the
radius, these develop a tachyon and become unstable when the $T$-dual
radius $\tilde{R}>\sqrt{2}$. For such radii the wrapped $D(p+1)$-branes
decay into unwrapped $Dp$-$\overline{Dp}$ systems constrained to the
$O8^+$-planes. The non-trivial $\bZ_2$ Wilson line in the type I theory
corresponds to the brane and anti-brane being on different
$O$-planes in the type IA theory. When a $\bZ_2$ charged
wrapped $D(p+1)$-brane decays in the type IA theory, its charge then
corresponds to the $\bZ_2$ choice of which $O$-plane the $Dp$-brane is
located on. 

In the region of stability of the type IA theory ($\tilde{R}<\sqrt{2}$) there
would, at first glance, seem to be more $\bZ_2$ charges than predicted
by $K$-theory. Given the above discussion we would expect the
$D0$-brane in the type IA theory to get a $\bZ_2$ charge contribution
from the choice of which $O$-plane to locate an unwrapped $D0$-brane
and another $\bZ_2$ charge contribution coming from a wrapped
$D1$-brane (since we are in the region of stability). However,
$K$-theory predicts that there should be only one source of $\bZ_2$
$D0$-brane charge. To understand this, consider a stuck $D0$-brane
(half of a $D0$-brane) at one $O$-plane and a wrapped $D1$-brane. This
has the same conserved charges as a stuck $D0$-brane at the other
$O$-plane and will decay into the latter configuration. In general a
stuck $Dp$-brane at one $O$-plane will be transferred to a stuck
$Dp$-brane at the other $O$-plane by a wrapped non-BPS
$D(p+1)$-brane. This is described in \cite{Hyakutake:2000} as a
non-BPS brane stretched between two $O$-planes switching the type of
$O$-plane between an $Op$- and $\widetilde{Op}$-plane (an
$\widetilde{Op}$-plane can be interpreted as an $Op$-plane with a
stuck $Dp$-brane). The above brane transfer operation shows that the
$\bZ_2$ $D0$-brane charge coming from the wrapped $D1$-brane and the
contribution from the choice of which $O$-plane the unwrapped
$D0$-brane is located at are not distinct sources of charge and that
the $K$-theory prediction that there is only one distinct source of
$\bZ_2$ $D0$-brane charge is correct. We have seen that the charge
spectrum remains unchanged in and out of the region of stability for
the non-BPS branes and that $K$-theory accurately classifies the
charges for the entire moduli space. 

Now let us return to the situation of interest, where we compactify
$2$ dimensions. The $Dp$-brane charges in the type IIB theory on $S^{2,0}\times S^{2,0}$ are classified by 
\begin{equation}
KR(S^{8-p,0}\times S^{2,0}\times S^{2,0},S^{2,0}\times S^{2,0})\cong KO^{p-7}\oplus 2KO^{p-8}\oplus KO^{p-9}.
\end{equation}
The $KO^{p-7}$ term corresponds to $Dp$-brane charge coming from
$D(p+2)$-branes wrapping the entire $\bT^2$. The $2$ copies of
$KO^{p-8}$ correspond to $D(p+1)$-branes wrapping the $2$ different
circles of $\bT^2$ and the $KO^{p-9}$ term corresponds to unwrapped
$Dp$-branes. The complete brane content is given in Table
\ref{Table:species2_charges}.  

The $Dp$-brane charges in the type IIA theory on $S^{2,0}\times S^{1,1}$ are classified by 
\begin{equation}
KR^{-1}(S^{8-p,0}\times S^{2,0}\times S^{1,1},S^{2,0}\times S^{1,1})\cong KO^{p-7}\oplus 2KO^{p-8}\oplus KO^{p-9}.
\end{equation}
Now, after performing a $T$-duality from the previous theory, the
$KO^{p-7}$ term corresponds to $Dp$-brane charge coming from
$D(p+1)$-branes wrapping $S^{2,0}$. One of the copies of
$KO^{p-8}$ now corresponds to a $D(p+2)$-branes wrapping $S^{2,0}\times S^{1,1}$, while the other corresponds to unwrapped
$Dp$-branes. The $KO^{p-9}$ term corresponds to $D(p+1)$-branes
wrapping $S^{1,1}$. 

The $Dp$-brane charges in the type IIB theory on $S^{1,1}\times S^{1,1}$ are classified by 
\begin{equation}
KR^{-2}(S^{8-p,0}\times S^{1,1}\times S^{1,1},S^{1,1}\times S^{1,1})\cong KO^{p-7}\oplus 2KO^{p-8}\oplus KO^{p-9}.
\end{equation}
Note the shift in index by $2$ from the relevant $KR$-theory for 
type I on $\bT^2$. Performing $2$ $T$-dualities shifts the index by
$2$. This fact is often overlooked when describing ordinary type
IIA/type IIB $T$-dualities on smooth manifolds. If we perform $2$
$T$-duality transformations on the type IIB theory on $S^1\times S^1$
we get back the type IIB theory on $S^1\times S^1$. $D$-brane charges
in the original theory are classified by $K(\bT^2)$; in the dual
theory they are classified $K^{-2}(\bT^2)$.  $K^{-2}(\bT^2)\cong
K(\bT^2)$ by Bott periodicity, but it is important to keep track of
the index shift for determining the dimensions of the branes
contributing the various charges. In our present discussion it is even
more important because the relevant $KR$-theory has period $8$ and not
$2$. For the type IIB theory on $S^{1,1}\times S^{1,1}$, the
$KO^{p-7}$ term corresponds to $Dp$-brane charge coming from unwrapped
$Dp$-branes. The $2$ copies of $KO^{p-8}$ correspond to
$D(p+1)$-branes wrapping the different copies of $S^{1,1}$ and the
$KO^{p-9}$ term corresponds to $D(p+2)$-branes wrapping $S^{1,1}\times
S^{1,1}$. 

It is pointed out in \cite{Dabholkar:1996} that performing $2$
$T$-dualities from the type I theory does not just lead to reflection
of both compact directions, but should also be combined with the
action of $(-1)^{F_L}$, where $F_L$ is the left-moving spacetime
fermion number. As described in \cite{Witten:1998}, $D$-branes in
orientifolds of the type $X/(\iota\cdot\Omega\cdot(-1)^{F_L})$ are
classified by $KR_\pm$. Using the definition for $KR_\pm$ given in
\cite{Bergman:2001-0105}, $KR_\pm(X)\cong KR(X\times\bR^{2,0})$, we see
that 
\begin{equation}
\label{eqn:KRpm}
KR_\pm(S^{1,1}\times S^{1,1})\cong KR^{-2}(S^{1,1}\times S^{1,1}).
\end{equation}
In this way, we can write our $T$-duality relationship between the
type IIB theory on $S^{2,0}\times S^{2,0}$ and on $S^{1,1}\times
S^{1,1}$ entirely in terms of $KR_\pm$. 
\begin{eqnarray}
KR_\pm(S^{1,1}\times S^{1,1}) &\cong& KR^2_\pm(S^{2,0}\times S^{2,0}) \nonumber\\
&\cong&KH_\pm^{-2}(S^{2,0}\times S^{2,0}),
\end{eqnarray}
where the last line was included for the sake of completeness and to
point out that depending on the variant of $KR$-theory we choose we
can make the degree change go in either direction, but it will always
be a change of $2$. It is also interesting to note that in this
example we were able to avoid concerning ourselves with $KR_\pm$ and
the presence of the $(-1)^{F_L}$ action by taking the appropriate
degree shift for $KR$. $KR_\pm$ added no new information beyond using
the correct degree for $KR$-theory. 

\begin{table}[ht]
\centering
\noindent\makebox[\textwidth]{\begin{tabular}{|| c || c | c | c | c | c | c | c | c | c || m{2.25cm} | m{2.25cm} | m{2.25cm} ||}
\hline
$Dp$-brane & $D7$ & $D6$ & $D5$ & $D4$ & $D3$ & $D2$ & $D1$ & $D0$ & $D(-1)$ & type I on $S^{2,0}_{(1)}\times S^{2,0}_{(2)}$ & type IIA on $S_{(1)}^{2,0}\times S_{(2)}^{1,1}$ & type IIB on $S_{(1)}^{1,1}\times S_{(2)}^{1,1}$\\
\hline\hline
$KO^{p-7}(\pt)$ & $\bZ$ & $\bZ_2$ & $\bZ_2$ & $0$ & $\bZ$ & $0$ & $0$ & $0$ & $\bZ$ & $(p+2)$-brane wrapping $\bT^2$  & $(p+1)$-brane wrapping $S_{(1)}^{2,0}$ & unwrapped $p$-brane\\ \hline
$KO^{p-8}(\pt)$ & $\bZ_2$ & $\bZ_2$ & $0$ & $\bZ$ & $0$ & $0$ & $0$ & $\bZ$ & $\bZ_2$ & $(p+1)$-brane wrapping $S_{(1)}^{2,0}$ & $(p+2)$-brane wrapping $S_{(1)}^{2,0}\times S_{(2)}^{1,1}$ & $(p+1)$-brane wrapping $S_{(2)}^{1,1}$ \\ \hline
$KO^{p-8}(\pt)$ & $\bZ_2$ & $\bZ_2$ & $0$ & $\bZ$ & $0$ & $0$ & $0$ & $\bZ$ & $\bZ_2$ & $(p+1)$-brane wrapping $S_{(2)}^{2,0}$ & unwrapped $p$-brane & $(p+1)$-brane wrapping $S_{(1)}^{1,1}$ \\ \hline
$KO^{p-9}(\pt)$ & $\bZ_2$ & $0$ & $\bZ$ & $0$ & $0$ & $0$ & $\bZ$ & $\bZ_2$ & $\bZ_2$ & unwrapped $p$-brane & $(p+1)$-brane wrapping $S_{(2)}^{1,1}$ & $(p+2)$-brane wrapping $S_{(1)}^{1,1}\times S_{(2)}^{1,1}$ \\ \hline
\end{tabular}}
\smallskip
\caption{$D$-brane charges for all of the theories related to the type I theory with trivial $B$-field on $\bR^{8}\times\bT^2$ by $T$-duality.}
\label{Table:species2_charges}
\end{table}

Again, this description for wrapped and unwrapped branes is not valid
in the entire moduli space. Non-BPS $D$-branes are not stable for all
possible radii of the compact dimensions. For determining the non-BPS
brane relations under $2$ $T$-dualities between the $2$ type IIB
theories, it is possible to follow a similar argument as for
compactification on a single circle, but higher order brane
transfer operations need to be accounted for. Note that even in this
case, the regions of stability will not be clear. When performing a
single $T$-duality to the type IIA theory, the situation becomes more
ambiguous. 

It is unclear if the stability conditions on a single circle can be
taken individually for the $2$ circles we now have, or if there is
some mixing. Additionally, the effect on the branes will depend on the
circle we $T$-dualize. As an example, consider the non-BPS $D0$-brane
in the type I theory compactified on $\bT^2$. For compactification on
a single circle, we saw that the $D0$-brane will decay into a
$D1$-$\overline{D1}$ pair wrapping the circle when the radius becomes
too small. An immediate question that arises when going to $2$ compact
dimensions is: if the radius of one circle gets too small will the
$D0$-brane decay or if the other circle has a large enough radius will
it be stable? Another possibility is that when the volume gets too
small, the $D0$-brane will decay into a $D1$-$\overline{D1}$ pair
wrapping the diagonal. Once we decide which circle the decay
$D1$-$\overline{D1}$ pair wraps, we have to consider which circle we're
$T$-dualizing. If we $T$-dualize the circle the $D1$-$\overline{D1}$
pair wrap, then they will map to an unwrapped $D0$-$\overline{D0}$ pair
located at the different $O$-planes. This corresponds to a situation
where the non-BPS $D1$-brane in the type IIA theory decays into
$D0$-$\overline{D0}$ pair. If, however, we $T$-dualize the circle
orthogonal to the $D1$-$\overline{D1}$ pair instead, they will map to a
wrapped $D2$-$\overline{D2}$ pair in the type IIA theory.\footnote{For the
  sake of completeness, we note that if the $D1$-$\overline{D1}$ pair
  wraps the diagonal, $T$-dualizing either leg will lead to a wrapped
  $D2$-$\overline{D2}$ pair with $B$-field.} In this case the non-BPS
$D1$-brane in the type IIA theory decays into a wrapped
$D2$-$\overline{D2}$ pair. 

Most likely all three of these objects: the non-BPS $D1$-brane,
the $D0$-$\overline{D0}$ pair, and the $D2$-$\overline{D2}$ pair, are all
stable in different regions of the moduli space, while in regions
where more than one is stable, brane transfer operations likely show
that they are not distinct sources of $D0$-brane charge. This does,
however, illustrate an important limitation of the $K$-theoretic
description of brane charges. The $K$-theory can only tell us there is
a stable source of non-BPS $D0$-brane charge. It can not determine
what that source is, let alone its regions of stability. To determine
this information, one would have to do a full boundary state
analysis. Our gained knowledge from the $K$-theoretic analysis does
greatly constrain what boundary states must be looked at, showing the
benefit of performing the $K$-theoretic analysis first. 

\subsection{The $T$-duality group defined on $y^2=-(1+x^2)(1+k^2x^2)$}

The series of $T$-dualities involving the elliptic curve
$y^2=-(1+x^2)(1+k^2x^2)$ follows a pattern very similar to the
previous case, but 
involves the type \~I and $\widetilde{IA}$ theories. Therefore, we
will review these two theories and their relation to one another
first. The full brane content is given in Table \ref{Table:charges_It}.

$Dp$-brane charges in the type \~I theory are classified by 
\begin{equation}
KR(S^{9-p,0}\times S^{0,2},S^{0,2})\cong KSC^{p-8}.
\end{equation}
$KSC$ doesn't split into pieces from wrapped and unwrapped
branes as in the 
species $2$ case. The authors of \cite{Bergman:1999} were still able
to determine which 
charges come from wrapped and unwrapped branes using what we know
about $T$-duality, the type IA theory, and $O8^\pm$-planes. We will
follow their argument here.  

As described in \cite{Doran:2013sxa}, $Dp$-brane charges in the type
$\widetilde{IA}$ theory are classified by 
$$KR^{-1}_{(+,-)}(\bR^{9-p,0}\times S^{1,1}, S^{1,1})=KSC^{p-8}.$$ 
(This will become clearer as we
explore the stability of $D$-branes near the $O8^+$ and
$O8^-$-planes.) We saw in the previous section that unwrapped
$Dp$-branes near an $O8^+$-plane are classified by $KO^{p-8}$ (see
Table \ref{Table:species2_S1}). 

Conversely, $O8^-$-planes are quantized with symplectic gauge bundles,
classified by $KSp(X)=KO^{-4}(X)$. Therefore, unwrapped $Dp$-brane
charges near the $O8^-$-plane are classified by $KSp^{p-8}$. See Table
\ref{Table:charges_O8pm} for a list of unwrapped $Dp$-brane charges
near $O8^\pm$-planes in a type IIA orientifold. 

\begin{table}[ht]
\centering
\noindent\makebox[\textwidth]{\begin{tabular}{|| c || c | c | c | c | c | c | c | c | c | c ||}
\hline
$Dp$-brane & $D8$ & $D7$ & $D6$ & $D5$ & $D4$ & $D3$ & $D2$ & $D1$ & $D0$ & $D(-1)$\\
\hline
$KO^{p-8}$ & $\bZ$ & $\bZ_2$& $\bZ_2$& $0$& $\bZ$& $0$& $0$& $0$& $\bZ$& $\bZ_2$\\ \hline
$KSp^{p-8}$ & $\bZ$& $0$& $0$& $0$& $\bZ$& $\bZ_2$& $\bZ_2$& $0$& $\bZ$& $0$\\ \hline
\end{tabular}}
\smallskip
\caption{Unwrapped $D$-brane charges near the $O8^+$- and $O8^-$-planes in a type IIA orientifold.}
\label{Table:charges_O8pm}
\end{table}

Let us first consider BPS branes. Table \ref{Table:charges_It} shows
BPS $D8$-branes, but tadpole cancellation in the type $\widetilde{IA}$
theory will require the net $D8$-brane charge to be zero. Table
\ref{Table:charges_It} shows that there are unwrapped BPS
$Dp$-branes for $p=0,4$. The $\bZ$ contribution to these charges
coming from both $KO$ and $KSp$ (see Table \ref{Table:charges_O8pm})
are equated. This is because $2$ half 
$D0$-branes on the $O8^+$-plane form a $D0$-brane in the bulk which
can then be interpreted as a $D0$-brane on the $O8^-$-plane. Similarly a $D4$-brane on the $O8^+$-plane can be considered as $2$
half $D4$-branes on the $O8^-$-plane. Now half $Dp$-branes can only
live on one of the $O8$-planes and it no longer makes sense to have a
brane transfer operation. Therefore there is no longer the $\bZ_2$
charge contribution coming from a choice of $O8$-plane that we saw in
the type IIA theory on $S^{1,1}$. It is important to note that the BPS
branes are stable near both $O^-$-planes and $O^+$-planes as is
apparent from the fact that there are integral contributions coming
from both the $KO$ and $KSp$ terms. 

We will now consider the unwrapped non-BPS $Dp$-branes. Unlike the
BPS case, we cannot just look locally at stable branes near the
different $O8$-planes, but need to take into account global
aspects. Table \ref{Table:charges_O8pm} correctly predicts the $\bZ_2$
charge contribution coming from unwrapped non-BPS $Dp$-branes for
$p=-1,3,7$, but it also seems to predict $\bZ_2$ charge contributions
from unwrapped non-BPS $Dp$-branes for $p=2,6$ that don't appear in
Table \ref{Table:charges_It}. This is because the $D6$-brane
($D2$-brane) is stable at the $O8^+$-plane ($O8^-$-plane), but
unstable near the $O8^-$-plane ($O8^+$-plane), so not globally
stable. Table \ref{Table:charges_It} lists only those charges that are
globally stable. Let us look a little closer at why the $D6$-brane is
not globally stable. The $D6$-brane can be viewed as $D6$-brane
together with its mirror $\overline{D6}$-brane. Near the $O8^+$-plane
the orientifold action projects out the tachyon in the system. Near
the $O8^-$-plane the projection is different and the tachyon is not
removed. 

We will now look at the wrapped brane charges following
\cite{Bergman:1999}, which determines the wrapped brane charges in the
type $\widetilde{IA}$ by considering unwrapped brane charges in the
type \~{I} theory. After going through their description, we will go
back and see how we can follow an argument similar to the one we used for the
unwrapped branes, by using the appropriately shifted $KO$ and $KSp$
groups. Wrapped $D(p+1)$-branes in type $\widetilde{IA}$ theory
correspond to unwrapped $Dp$-branes in the type \~{I} theory, so we
will consider unwrapped branes in the type \~{I} theory. Since
$D$-branes in the type \~{I} theory must obey the symmetry we modded
the type IIB theory on $\bR^9\times S^1$ out by (which includes a
rotation of $S^1$ by $\pi$ radians), we must equate an unwrapped
$Dp$-brane with another $Dp$-brane at the opposite point on the circle
for $p=1,5$ and a $\overline{Dp}$-brane for $p=-1,3,7$. The $D1$- and
$D5$-brane configurations are stable and contribute the BPS $D1$- and
$D5$-brane charges appearing in table \ref{Table:charges_It}. They
correspond to $D2$- and $D6$-branes in the type $\widetilde{IA}$
theory that wrap the compact dimension twice respectively. For
$p=-1,3,7$ the $Dp$-$\overline{Dp}$ systems give stable non-BPS
states. To see that these states carry $\bZ_2$ charge, consider a
system consisting of two such states. While each individual
$Dp$-$\overline{Dp}$ pair at opposite points of the circle is stable, the $Dp$-brane from one state can
annihilate with the $\overline{Dp}$-brane from the other state and
\emph{vice versa}. This would seem to imply $2$ sources of $\bZ_2$ charge in
the type \~{I} theory; one from the unwrapped $Dp$-branes with
$p=-1,3,7$ just described and the other from wrapped $D(p+1)$-branes
$p=-1,3,7$ corresponding to the unwrapped $Dp$-branes with $p=-1,3,7$
in the type $\widetilde{IA}$ theory via $T$-duality. The $K$-theory,
however, predicts that there should only be one source of $\bZ_2$
$Dp$-brane charge for $p=-1,3,7$. This is because the two different
types of states (wrapped and unwrapped branes) are stable in different
regions of the moduli space. In the type \~{I} theory the unwrapped
$Dp$-$\overline{Dp}$ pair are stable for large $R$, while the wrapped
$D(p+1)$-brane is stable for small $R$. 

\begin{table}[h]
\centering
\noindent\makebox[\textwidth]{\begin{tabular}{|| c || c || c | c | c ||}
\hline
$Dp$-brane & $KSC^{p-8}$ & Region of Stability & Type \~I & Type $\widetilde{IA}$\\
\hline\hline
$D8$ & $\bZ$ & stable for all radii & wrapped $D9$-brane  & unwrapped $D8$-brane \\ \hline
\multirow{2}{*}{$D7$} & \multirow{2}{*}{$\bZ_2$}& $R_{\tilde{I}}<\frac{1}{\sqrt{2}}, R_{\widetilde{IA}}>\sqrt{2}$  & wrapped $D8$-brane & unwrapped $D7$-brane \\ \cline{3-5}
& & $R_{\tilde{I}}>\frac{1}{\sqrt{2}}, R_{\widetilde{IA}}<\sqrt{2}$ & unwrapped $D7$-brane & wrapped $D8$-brane\\ \hline
$D6$ & $0$& &  & \\ \hline
$D5$ & $\bZ$& stable for all radii & unwrapped $D5$-brane  & doubly wrapped $D6$-brane \\ \hline
$D4$ & $\bZ$& stable for all radii  & wrapped $D5$-brane  & unwrapped $D4$-brane \\ \hline
\multirow{2}{*}{$D3$} & \multirow{2}{*}{$\bZ_2$}& $R_{\tilde{I}}<\frac{1}{\sqrt{2}}, R_{\widetilde{IA}}>\sqrt{2}$  & wrapped $D4$-brane & unwrapped $D3$-brane \\ \cline{3-5}
& & $R_{\tilde{I}}>\frac{1}{\sqrt{2}}, R_{\widetilde{IA}}<\sqrt{2}$ & unwrapped $D3$-brane & wrapped $D4$-brane\\ \hline
$D2$ & $0$&  & & \\ \hline
$D1$ & $\bZ$& stable for all radii  & unwrapped $D1$-brane  & doubly wrapped $D2$-brane \\ \hline
$D0$ & $\bZ$& stable for all radii  & wrapped $D1$-brane  & unwrapped $D0$-brane \\ \hline
\multirow{2}{*}{$D(-1)$} & \multirow{2}{*}{$\bZ_2$}& $R_{\tilde{I}}<\frac{1}{\sqrt{2}}, R_{\widetilde{IA}}>\sqrt{2}$  & wrapped $D0$-brane & unwrapped $D(-1)$-brane \\ \cline{3-5}
& & $R_{\tilde{I}}>\frac{1}{\sqrt{2}}, R_{\widetilde{IA}}<\sqrt{2}$ & unwrapped $D(-1)$-brane & wrapped $D0$-brane\\ \hline
\end{tabular}}
\smallskip
\caption{$D$-brane charges in the type \~I and type $\widetilde{IA}$ theories.}
\label{Table:charges_It}
\end{table}

We could also have determined the wrapped branes in the type $\widetilde{IA}$
theory by looking at the appropriate $K$-theory in the vicinity of the
$O$-planes. We saw in the previous subsection that wrapped $Dp$-brane
charge in the type IA theory is classified by $KO^{p-9}$, so this will
classify wrapped $Dp$-brane charges near the $O8^+$-plane. Near the
$O8^-$-plane the orthogonal bundle is replaced with a symplectic
bundle, so the wrapped $Dp$-brane charge will similarly be classified
by $KSp^{p-9}$.  
\begin{table}[htb]
\centering
\noindent\makebox[\textwidth]{\begin{tabular}{|| c || c | c | c | c | c | c | c | c | c | c ||}
\hline
$Dp$-brane & $D8$ & $D7$ & $D6$ & $D5$ & $D4$ & $D3$ & $D2$ & $D1$ & $D0$ & $D(-1)$\\
\hline
$KO^{p-9}$ & $\bZ_2$& $\bZ_2$& $0$& $\bZ$& $0$& $0$& $0$& $\bZ$& $\bZ_2$& $\bZ_2$\\ \hline
$KSp^{p-9}$ & $0$& $0$& $0$& $\bZ$& $\bZ_2$& $\bZ_2$& $0$& $\bZ$& $0$& $0$\\ \hline
\end{tabular}}
\smallskip
\caption{Wrapped $D$-brane charges near the $O8^+$- and $O8^-$-planes in type IIA orientifolds.}
\label{Table:charges_O8pmw}
\end{table}
As can be seen from table \ref{Table:charges_O8pmw}, this correctly
accounts for the BPS $D5$- and $D1$-brane charge coming from wrapped
$D6$- and $D2$-branes respectively. It also correctly predicts the
non-BPS $Dp$-brane charge contribution from wrapped
$D(p+1)$-branes for $p=-1,3,7$. It would also seem to imply the
existence of non-BPS $Dp$-brane charge for $p=0,4,8$ coming from wrapped
$D(p+1)$-branes. Just as in the unwrapped case, the non-BPS wrapped
$D5$-brane, for example, will be stable near the $O8^+$-plane, but not
globally stable. This example shows how $KR$-theory picks up all
global aspects of stable $D$-brane charges on the orientifold, though
the information about wrapped and unwrapped branes is sometimes
obscured. We were able to gain that information by looking at the
appropriate $K$-theory that classifies charges locally near each $O$-plane
and then comparing it to the $KR$-theory to see which locally stable
states are globally stable. In fact, with the hindsight of knowing
that non-BPS $Dp$-branes come in pairs (i.e., an unwrapped non-BPS $Dp$
brane will decay into a wrapped non-BPS  $D(p+1)$-brane for certain
radii), we can determine the $Dp$-brane charges by comparing the
wrapped and unwrapped spectrum. 

We see by comparing the first and second lines of Table
\ref{Table:charges_O8pm} with the first and second lines of Table
\ref{Table:charges_O8pmw} that the stable non-BPS $D$-branes are those
that have stable charges in the first line of Table
\ref{Table:charges_O8pm} and the first line of Table
\ref{Table:charges_O8pmw}, or in the second line of Table
\ref{Table:charges_O8pm} and the second line of Table
\ref{Table:charges_O8pmw}. This corresponds to the fact that both the
wrapped and unwrapped brane must be stable since they contribute to
the charge in different regions of the moduli space. 

Once all of the relevant $K$-theories are known, much of the $D$-brane
content can be determined from the long exact sequence 
\begin{equation}
\cdots\to KSC^{-n-1}(X)\to K^{-n}(X)\to KO^{-n}(X)\oplus
KSp^{-n}(X)\to KSC^{-n}(X)\to\cdots, 
\end{equation}
as suggested in \cite{Olsen:1999}.

As an example, consider the segment starting with $K^{-1}(\pt)$
\begin{equation}
\label{eqn:KSC_LES1}
\xymatrix{0\ar[r] & KO^{-1}\oplus KSp^{-1}\ar[r]\ar@{=}[d]&
  KSC^{-1}\ar^(.6)\alpha[r]\ar@{=}[d]& K\ar^(.35)\beta[r]\ar@{=}[d]&
  KO\oplus KSp\ar[r]\ar@{=}[d]& KSC\ar[r]\ar@{=}[d]& 0\\ 
0\ar[r] & \bZ_2\oplus 0\ar[r] & \bZ_2\ar[r] & \bZ\ar[r] &
\bZ\oplus\bZ\ar[r] & \bZ\ar[r] & 0.} 
\end{equation}
We see immediately that the $2$-torsion in $KSC^{-1}$ comes from
$KO^{-1}$. $KO^{-1}(\pt)$ corresponds to the $\bZ_2$ charge coming
from an unwrapped $D7$-brane or an unwrapped
$D(-1)$-brane in the type IIA theory. We can also see immediately that
$\alpha=0$ in equation \eqref{eqn:KSC_LES1}. This shows that $KSC(\pt)$
is $KO(\pt)\oplus KSp(\pt)$ modulo the relation equating the generators
of $KO$ and $KSp$. This corresponds to the BPS $D$-brane
charge coming from an unwrapped $D8$-brane, unwrapped $D4$-brane, and
unwrapped $D0$-brane. 

Again, this method cannot tell us anything about regions of stability,
or really anything about the sources. We were able to determine the
sources in this situation because of previous knowledge about the
relationship between the physical sources near $O8^\pm$-planes and
$KO$ and $KSp$. 

Now that we've reviewed the type \~I and $\widetilde{IA}$ theories we
can easily obtain the species $0$ cases we are interested in by
compactifying on another circle. 

It is easiest to first discuss the brane content for a half shift in
only one direction, as is pictured in Figure \ref{Fig:species0r}. This
is because Figure \ref{Fig:species0r} corresponds to compactifying the
type \~I and type $\widetilde{IA}$ theories on another
circle. $Dp$-brane charges in the type IIB on $S^{0,2}\times S^{2,0}$
are classified by 
\begin{equation}
KR(S^{8-p}\times S^{0,2}\times S^{2,0}, S^{0,2}\times S^{2,0})\cong KSC^{p-7}\oplus KSC^{p-8}.
\end{equation}

We can determine the $D$-brane content by compactifying the type \~I
theory on a copy of $S^{2,0}$. Now $Dp$-branes in the type \~I theory
can wrap $S^{2,0}$ and we see that $KSC^{p-8}$ classifies branes in
the type \~I theory that do not wrap $S^{2,0}$, while $KSC^{p-7}$
classifies branes from the type \~I theory that now wrap
$S^{2,0}$. For example, $D7$-brane charge is classified by $KSC\oplus
KSC^{-1}$. The $\bZ$ charge coming from $KSC$ corresponds to the
integral $D8$-brane charge from the type \~I theory now wrapping
$S^{2,0}$. Since the BPS $D8$-brane charge in the type \~I theory came
from a $D9$-brane wrapping $S^{0,2}$, wrapping it additionally on
$S^{2,0}$ shows that the BPS $D7$-brane charge comes from a $D9$-brane
wrapping the entire compact space. The $\bZ_2$ charge coming from
$KSC^{-1}$ corresponds to the $D7$-brane charge in the type \~I theory
that does not wrap $S^{2,0}$. For these branes there are stability
conditions (not present with BPS branes) that cannot be determined by
the $K$-theory analysis. 

We saw that for the type \~I theory, the non-BPS $D7$ brane charge
corresponds to an unwrapped $D7$-brane for large radius and a wrapped
$D8$-brane for small radius (see Table \ref{Table:charges_It}). For
one compactification circle, the stability of the $D7$-brane required
a large radius because in the covering circle the $D7$-brane is a
$D7$-$\overline{D7}$ pair located at antipodal points of the
circle. This argument continues to make sense when we compactify on an
additional circle; however, it is unclear how the stability of the
unwrapped $D7$-brane or wrapped $D8$-brane will depend on the radius
of $S^{2,0}$. Determining the non-BPS brane stability only
in terms of the size of the underlying type \~I theory,\footnote{This
  corresponds to assuming the radius of $S^{2,0}$ is large.} the full
brane content is given in Table \ref{Table:charges_species0r}. Determining
the full stability conditions for the non-BPS branes would again
require doing a full boundary state analysis. The brane content for
the other theories involved can be determined via $T$-duality and is also
shown in Table \ref{Table:charges_species0r}. 

\begin{table}[pht]
\centering
\noindent\scalebox{0.75}{\makebox[\textwidth]{
\begin{tabular}{|| c || c || m{2.25cm} | m{2.25cm} | m{2.25cm} |
    m{2.25cm} | m{2.25cm} ||} 

\hline
\multirow{2}{*}{$Dp$-brane} & $KSC^{p-7}(\pt)$ & Region of & Type IIB
on & Type IIA on & type IIA on &Type IIB on \\ 
\cline{2-2}
 & $KSC^{p-8}(\pt)$ & Stability  & $S_{R_1}^{0,2}\times S_{R_2}^{2,0}$  &
$S_{(+,-)}^{1,1}\times S^{2,0}$ &$S^{0,2}\times S^{1,1}$  &
$S_{(+,-)}^{1,1}\times S^{1,1}$ \\ 
\hline\hline

\multirow{3}{*}{$D7$} & $\bZ$ & stable for all radii & wrapped
$D9$-brane  & $D8$-brane wrapping $S^{2,0}$ & $D8$-brane wrapping
$S^{0,2}$ & unwrapped $D7$-brane \\ \cline{2-7}\cline{2-7} 
&\multirow{2}{*}{$\bZ_2$}& $R_1<\frac{1}{\sqrt{2}}$  & $D8$-brane
wrapping $S^{0,2}$ & unwrapped $D7$-brane  & wrapped $D9$-brane &
$D8$-brane wrapping $S^{1,1}$\\ \cline{3-7} 
& & $R_1>\frac{1}{\sqrt{2}}$ & unwrapped $D7$-brane & $D8$-brane
wrapping $S_{(+,-)}^{1,1}$ & $D8$-brane wrapping $S^{1,1}$ & wrapped
$D9$-brane \\  
\hline\hline

\multirow{3}{*}{$D6$} & \multirow{2}{*}{$\bZ_2$}&
$R_1<\frac{1}{\sqrt{2}}$  & wrapped $D8$-brane & $D7$-brane wrapping
$S^{2,0}$  & $D7$-brane wrapping $S^{0,2}$ & unwrapped
$D6$-brane\\ \cline{3-7} 
& & $R_1>\frac{1}{\sqrt{2}}$ & $D7$-brane wrapping $S^{2,0}$ & wrapped
$D8$-brane & unwrapped $D6$-brane & $D7$-brane wrapping
$S_{(+,-)}^{1,1}$ \\ \cline{2-7}\cline{2-7} 
& $0$ &   &  &  &  \\ 
\hline\hline

\multirow{2}{*}{$D5$} & $0$ & & & & &\\ \cline{2-7}\cline{2-7}
& $\bZ$& stable for all radii & unwrapped $D5$-brane  & $D6$-brane wrapping $S_{(+,-)}^{1,1}$ twice & $D6$-brane wrapping $S^{1,1}$  & wrapped $D7$-brane \\ 
\hline\hline

\multirow{2}{*}{$D4$} & $\bZ$ & stable for all radii & $D5$-brane
wrapping $S^{2,0}$ & doubly wrapped $D6$-brane & unwrapped $D4$-brane
& $D5$-brane wrapping $S_{(+,-)}^{1,1}$\\ \cline{2-7}\cline{2-7} 
& $\bZ$& stable for all radii & $D5$-brane wrapping $S^{0,2}$  &
unwrapped $D4$-brane & wrapped $D6$-brane & $D5$-brane wrapping
$S^{1,1}$ \\  
\hline\hline

\multirow{3}{*}{$D3$} & $\bZ$ & stable for all radii & wrapped
$D5$-brane  & $D4$-brane wrapping $S^{2,0}$ & $D4$-brane wrapping
$S^{0,2}$ & unwrapped $D3$-brane \\ \cline{2-7}\cline{2-7} 
&\multirow{2}{*}{$\bZ_2$}& $R_1<\frac{1}{\sqrt{2}}$  & $D4$-brane
wrapping $S^{0,2}$ & unwrapped $D3$-brane  & wrapped $D5$-brane &
$D4$-brane wrapping $S^{1,1}$\\ \cline{3-7} 
& & $R_1>\frac{1}{\sqrt{2}}$ & unwrapped $D3$-brane & $D4$-brane
wrapping $S_{(+,-)}^{1,1}$ & $D4$-brane wrapping $S^{1,1}$ & wrapped
$D5$-brane \\  
\hline\hline

\multirow{3}{*}{$D2$} & \multirow{2}{*}{$\bZ_2$}&
$R_1<\frac{1}{\sqrt{2}}$  & wrapped $D4$-brane & $D3$-brane wrapping
$S^{2,0}$  & $D3$-brane wrapping $S^{0,2}$  & unwrapped $D2$-brane
\\ \cline{3-7} 
& & $R_1>\frac{1}{\sqrt{2}}$ & $D3$-brane wrapping $S^{2,0}$ & wrapped
$D4$-brane & unwrapped $D2$-brane & $D3$-brane wrapping
$S_{(+,-)}^{1,1}$ \\ \cline{2-7}\cline{2-7} 
& $0$ &   &  &  &  \\ 
\hline\hline

\multirow{2}{*}{$D1$} & $0$ & & & & &\\ \cline{2-7}\cline{2-7}
& $\bZ$& stable for all radii & unwrapped $D1$-brane  & $D2$-brane
wrapping $S_{(+,-)}^{1,1}$ &$D2$-brane wrapping $S^{1,1}$ & wrapped
$D3$-brane \\  
\hline\hline

\multirow{2}{*}{$D0$} & $\bZ$ & stable for all radii &  $D1$-brane
wrapping $S^{2,0}$ & wrapped $D2$-brane & unwrapped $D0$-brane &
$D1$-brane wrapping $S_{(+,-)}^{1,1}$  \\ \cline{2-7}\cline{2-7} 
& $\bZ$& stable for all radii & $D1$-brane wrapping $S^{0,2}$  &
unwrapped $D0$-brane & wrapped $D2$-brane & $D1$-brane wrapping
$S^{1,1}$\\  
\hline\hline

\multirow{3}{*}{$D(-1)$} & $\bZ$ & stable for all radii & wrapped
$D1$-brane  & $D0$-brane wrapping $S^{2,0}$ & $D0$-brane wrapping
$S^{0,2}$ & unwrapped $D(-1)$-brane \\ \cline{2-7}\cline{2-7} 
&\multirow{2}{*}{$\bZ_2$}& $R_1<\frac{1}{\sqrt{2}}$  & $D0$-brane
wrapping $S^{0,2}$ & unwrapped $D(-1)$-brane  & wrapped $D1$-bane &
$D0$-brane wrapping $S^{1,1}$\\ \cline{3-7} 
& & $R_1>\frac{1}{\sqrt{2}}$ & unwrapped $D(-1)$-brane & $D0$-brane
wrapping $S^{0,2}$ & $D0$-brane wrapping $S^{1,1}$ & wrapped
$D1$-brane \\  
\hline\hline

\end{tabular}}}
\smallskip
\caption{$D$-brane charges in the type \~I and type $\widetilde{IA}$
  theories assuming $R_2$ large.} 
\label{Table:charges_species0r}
\end{table}

The $Dp$-brane charges in the type IIA theory compactified on
$S^{0,2}\times S^{1,1}$ are classified by 
\begin{equation}
KR^{-1}(S^{8-p}\times S^{0,2}\times S^{1,1}, S^{0,2}\times
S^{1,1})\cong KSC^{p-7}\oplus KSC^{p-8}. 
\end{equation}
Here, $KSC^{p-7}$ classifies brane that don't wrap $S^{1,1}$ and
$KSC^{p-8}$ classifies branes that do wrap $S^{1,1}$, since this
theory is obtained from the IIB theory on $S^{0,2}\times S^{2,0}$ by
$T$-dualizing $S^{2,0}$. The complete brane content is listed in Table
\ref{Table:charges_species0r}. 

For the type IIA theory on $S_{(+,-)}^{1,1}\times S^{2,0}$ there had been no
description of the brane content in terms of the $KR$-theory of the
topological compactification space, $S^{1,1}\times S^{2,0}$. This led
us to define $KR$-theory with a sign choice in
\cite{Doran:2013sxa}. $Dp$-brane charges are classified by 
\begin{equation}
KR^{-1}_{(+,-)}(S^{8-p}\times S^{2,0}\times S^{1,1}, S^{2,0}\times
S^{1,1})\cong KSC^{p-7}\oplus KSC^{p-8}. 
\end{equation}
Here $KSC^{p-7}$ classifies branes that wrap $S^{2,0}$,
$KSC^{p-8}$ classifies branes that don't wrap $S^{2,0}$, and the
branes that wrap $S_{(+,-)}^{1,1}$ are the same as those that wrap
$S^{1,1}$ in the type IIA theory on $S^{0,2}\times S^{1,1}$ by $2$ $T$-dualities. 

Finally, $Dp$-branes in the type
IIB theory on $S_{(+,-)}^{1,1}\times S_{(+,+)}^{1,1}$ are classified by
\begin{equation}
KR^{-2}_{(+,+,-,-)}(S^{8-p}\times S^{1,1}\times S^{1,1}, S^{1,1}\times
S^{1,1})\cong KSC^{p-7}\oplus KSC^{p-8}, 
\end{equation}
with $KSC^{p-7}$ corresponding to branes
that don't wrap $S^{1,1}$ and $KSC^{p-8}$ corresponding to branes that
do. 

Now let's turn our attention to the case where we shift both the real
and imaginary directions by a half (Figure \ref{Fig:species0b}). The $2$
type IIA theories occurring in Figure \ref{Fig:species0b} are
$S_{(+,-)}^{1,1}\times S^{0,2}$ and $S^{0,2}\times S_{(+,-)}^{1,1}$,
which are dianalytically equivalent to $S^{1,1}\times S^{0,2}$ and
$S^{0,2}\times S^{1,1}$, respectively. Therefore the $Dp$-brane charges
are classified by $KR^{-1}(S^{8-p}\times S^{0,2}\times S^{1,1},
S^{0,2}\times S^{1,1})$. The type IIB theory with $4$ fixed points is
the same as before, but we have introduced a new ambiguity for the
type IIB theory with no fixed points. 

As an example, consider the non-BPS $D7$ charge in the type IIB theory
with no fixed points. We saw that when we shifted in one direction, the
source for this charge was a $D8$-brane wrapping $S^{0,2}$ (at least in
some region of the moduli space). $S^{0,2}\times S^{0,2}$ is
topologically equivalent to $S^{0,2}\times S^{2,0}$, so we would
expect the $D$-brane content to be the same. There are now $2$ copies
of $S^{0,2}$, however, so it is no longer immediately clear which one the
$D8$-brane should wrap. This is related to determining what direction
we should $T$-dualize in as was discussed earlier.  

As noted, $Dp$-brane charges in the type IIA theory on
$S_{(+,-)}^{1,1}\times S^{0,2}$ are classified by
\[
KR^{-1}(S^{8-p}\times S^{0,2}\times S^{1,1}, S^{0,2}\times S^{1,1}),
\]
which is the same as for the non-symmetric case. For the non-symmetric
case, the double $T$-duality between the $2$ type IIA theories related
$2$ different theories. For the symmetric case, it relates the same
theory. 

If $R_1$ and $R_2$ are both large (or both small large) then the $2$
IIA theories in the symmetric case will be in the same regions of
stability for the non-BPS branes. For concreteness, consider the case
where $R_1$ and $R_2$ are both large. The brane content under this
assumption for the non-symmetric case is given in Table
\ref{Table:charges_species0r}. In this region the non-BPS brane charge
comes from an unwrapped $D7$-brane in the type IIB theory on
$S^{0,2}\times S^{0,2}$. In the $2$ T-dual IIA theories this comes
from a $D8$-brane wrapping the copy of $S^{1,1}_{(+,-)}$ in
$S^{1,1}_{(+,-)}\times S^{0,2}$ or $S^{0,2}\times S^{1,1}_{(+,-)}$. In
both cases $S^{1,1}_{(+,-)}$ has a small radius and $S^{0,2}$ has a
large radius, so the IIA theories are truly symmetric. If $R_1$ is
small and $R_2$ large (or vice versa), then the two IIA theories are in
different regions of the moduli space. So if we start with 
the type IIA theory on $S_{(+,-)}^{1,1}\times S^{0,2}$ where the torus
has small volume, the double $T$-dual will give the type IIA theory on
$S^{0,2}\times S_{(+,-)}^{1,1}$ where the torus has large
volume. According to Table \ref{Table:charges_species0r} we would
expect the non-BPS $D7$-charge to be given by a wrapped $D9$-brane in
the type IIA theory on $S^{1,1}_{(+,-)}\times S^{0,2}$ when both
compact directions have small radii. Under a double $T$-duality we
would expect the non-BPS $D7$-brane charge to come from an unwrapped
$D7$-brane in the type IIA theory on $S^{0,2}\times S^{1,1}_{(+,-)}$
with large volume. It is reasonable to expect the unwrapped $D7$-brane
to be stable for large volume based on what we know about type IIA
circle orientifolds, but without performing a full boundary state
analysis we 
cannot be sure how the stability conditions for non-BPS branes on
$S^{0,2}$ and $S_{(+,-)}^{1,1}$ combine in our current case. 

We cannot extend the results in Table \ref{Table:charges_species0r} to
the symmetric case following the prescription described above for
$D6$- and $D2$-brane charge. Let us consider the case of $D6$-brane
charge. For the type IIA theory on $S^{0,2}\times S^{1,1}_{(+,-)}$, we
would expect the non-BPS $D6$-brane charge to come from an unwrapped
$D6$-brane when $S^{0,2}$ has a large radius and $S^{1,1}_{(+,-)}$ has
a small radius, by comparison to the non-symmetric case. This, however,
does not make sense. Under $2$ $T$-dualities this would map to a
$D8$-brane wrapping $S^{1,1}_{(+,-)}\times S^{0,2}$, where again
$S^{0,2}$ has a large radius and $S^{1,1}_{(+,-)}$ has a small
radius. This would imply that both unwrapped $D6$-branes and wrapped
$D8$-branes are stable (and dependent) sources of non-BPS $D6$-brane
charge in this region of the moduli space of $S^{0,2}\times
S^{1,1}_{(+,-)}$. We could then expect the wrapped $D8$-brane to be
stable in $S^{0,2}\times S^{1,1}$. By $T$-duality this would imply a
stable unwrapped $D6$-brane in the type IIA theory on
$S^{1,1}_{(+,-)}\times S^{2,0}$. We know this cannot be possible since
$D6$-branes are unstable near $O^-$-planes in type IIA theories. The
problem is seen more easily by noting that if an unwrapped $D6$-brane
was stable in the type IIA theory on $S^{0,2}\times S^{1,1}_{(+,-)}$,
then under $T$-duality there would be a stable $D7$-brane wrapping
$S^{0,2}$ in the type IIB theory on $S^{0,2}\times S^{0,2}$, which is
not possible. The problem for both the $D2$ and $D6$ charges is that
in the non-symmetric case there is a region where the charge comes
from a $D(p+1)$-brane, $p=2,6$, wrapping $S^{2,0}$ which is not stable
wrapping $S^{0,2}$ (see Table \ref{Table:charges_It}). 

One possible solution to this is simply to say that the only source
for non-BPS $D6$-brane charge in the type IIB theory of $S^{0,2}\times
S^{0,2}$ is a wrapped $D8$-brane, but there are several unsatisfactory
consequences of this. This would preclude the possibility of a stable
$D7$-brane wrapping $S^{1,1}_{(+,-)}$ in the type IIB theory with $4$
fixed points and assume the unwrapped $D6$-brane is stable everywhere
in the moduli space. We would expect the unwrapped $D6$-brane in the
type IIB theory on $S^{1,1}_{(+,-)}\times S^{1,1}_{(+,-)}$ to be
unstable for small volume and a $D7$-brane wrapping $S^{1,1}_{(+,-)}$
to be stable there (more on this below). Furthermore, we know there is
a copy of $S^{2,0}$ in $S^{0,2}\times S^{0,2}$ from Figure
\ref{Fig:species0twist}.  

As another possible resolution to the sources of $D6$-brane charge,
consider the theory with involution $z\mapsto
-\bar{z}+\frac{1+\tau}{2}$ with $\tau=it$, the type IIA theory on
$S^{1,1}_{(+,-)}\times S^{0,2}$. We mentioned earlier that in the $T$-dual
theory $S^{0,2}\times S^{0,2}$, branes that wrap $S^{0,2}$ should wrap
the diagonal since it is equivariant. The diagonal is no longer
equivariant in $S^{1,1}_{(+,-)}\times S^{0,2}$ --- it is exchanged with
the antidiagonal.  Instead we should consider pairs of branes that
wrap the equivariant copies of $S^{0,2}$ pictured in Figure
\ref{Fig:1wrapspecies0sym}. The reason we need to consider pairs
should become apparent momentarily. Note that the red circle and green
circle each wrap the imaginary direction corresponding to $S^{0,2}$,
but do not wrap the real direction corresponding to
$S^{1,1}_{(+,-)}$. It is hard to see what happens to this pair of
branes under $T$-duality, but notice that we can decompose them as the
diagonal and antidiagonal. 

\begin{figure}[h]
\centering
\begin{tikzpicture}[baseline]
\draw[black] (-1,-1) -- (1,-1);
\draw[black] (-1,1) -- (1,1);
\draw[black] (-1,-1) -- (-1,1);
\draw[black] (1,-1) -- (1,1);
\draw[red, very thick] (-1,-1) -- (0,0);
\draw[green] (0,0) -- (1,-1);
\draw[red, very thick] (-1,1) -- (0,0);
\draw[green] (0,0) -- (1,1);
\draw[fill,red] (-1,-1) circle[radius=0.0625];
\draw[fill,red] (-1,1) circle[radius=0.0625];
\draw[fill,red] (1,-1) circle[radius=0.0625];
\draw[fill,red] (1,1) circle[radius=0.0625];
\draw[fill,blue] (-0.05,-0.05) rectangle +(0.1,0.1);
\end{tikzpicture}
\caption{The red thick and green thinner lines show equivariant copies
  of $S^{0,2}$ in $S^{1,1}_{(+,-)}\times S^{0,2}$, which intersect at
  the blue square and red circles (note that all red circles are
  equated).} 
\label{Fig:1wrapspecies0sym}
\end{figure}
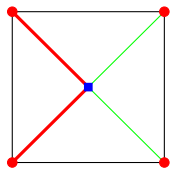

If we then $T$-dualize in the imaginary direction to get
$S^{0,2}\times S^{0,2}$, branes that wrap the diagonal and
anti-diagonal will map to branes that wrap the real direction; see
Figure \ref{Fig:S1insymspecies0}. 

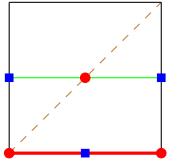
\begin{figure}[h]
\centering
\begin{tikzpicture}[baseline]
\draw[red, very thick] (-1,-1) -- (1,-1);
\draw[brown,dashed] (-1,-1) -- (1,1);
\draw[black] (-1,1) -- (1,1);
\draw[black] (-1,-1) -- (-1,1);
\draw[black] (1,-1) -- (1,1);
\draw[green] (-1,0) -- (1,0);
\draw[fill,red] (-1,-1) circle[radius=0.0625];
\draw[fill,red] (1,-1) circle[radius=0.0625];
\draw[fill,red] (0,0) circle[radius=0.0625];
\draw[fill,blue] (-1.05,-0.05) rectangle +(0.1,0.1);
\draw[fill,blue] (0.95,-0.05) rectangle +(0.1,0.1);
\draw[fill,blue] (-0.05,-1.05) rectangle +(0.1,0.1);
\end{tikzpicture}
\caption{The red and green lines show a copy $S^{2,0}$ in
  $S^{0,2}\times S^{0,2}$, relative to $\widetilde{\tau}=\tau+1$,
  shown by the brown dashed line.} 
\label{Fig:S1insymspecies0}
\end{figure}

As described in the discussion of Figure \ref{Fig:species0twist}, the
pair of red and green lines in Figure \ref{Fig:S1insymspecies0}
($T$-dual to the pair of red and green lines in Figure
\ref{Fig:1wrapspecies0sym}) together define a copy of $S^{2,0}$ in
$S^{0,2}\times S^{0,2}$ relative to the equivalent complex modulus
$\tilde{\tau}=\tau+1$. To describe the unwrapped $D6$- and wrapped
$D8$-branes that appear in the non-symmetric case, we rely on the
previous observation that the non-symmetric case can be obtained from
the symmetric case by instead $T$-dualizing in the $\tilde{\tau}$
direction. Note that when $T$-dualizing along the diagonal, a brane
that wraps to the diagonal will map to an unwrapped brane, while a
brane that wraps the antidiagonal will map to a wrapped brane. While
we cannot give the sources for all of the non-BPS charges, the
$K$-theory analysis greatly constrains what boundary states need to be
considered. 

Finally, we make one last note about the $D$-branes in the type IIB
theory with involution $z\mapsto -z+\frac{1+\tau}{2}$. This is the
same as the type IIB theory with involution $z\mapsto-z+\frac{1}{2}$
(which we considered previously), with the only difference being which
$2$-torsion points are exchanged. The exchange of $2$-torsion points
corresponds to an $O^+$-$O^-$-plane pair, so the only difference
between the two theories is the relative location of the $O^+$- and
$O^-$-planes. Therefore we can easily convert our previous discussion
of $D$-brane content. For example, in a certain region of the moduli
space we found there was a stable $D8$-brane wrapping $S^{1,1}$. In
general this corresponds to a $D8$-brane stretched between the $2$
$O^+$-planes. 

\subsection{The $T$-duality group defined on $y^2=(1-x^2)(1-k^2x^2)$, $k^2<0$}

Letting $M$ be any of the species $1$ real elliptic curves, $Dp$-brane
charges in the type IIA theory with species $1$ are classified by 
\begin{equation}
KR^{-1}(S^{8-p,0}\times M, M)\cong KR^{p-8}(M).
\end{equation}
The calculation of these $KR$-groups is given in Section
\ref{sec:KRcalcs} and results in terms of $D$-brane charges are given
in Table \ref{Table:charges}.

The $KR$-groups do not split into wrapped and unwrapped terms as in
the previous $2$ cases. Before discussing what we can determine about
the sources, let us briefly discuss the charge classifications in the
type IIB theories. $Dp$-brane charges in the type I theory without
vector structure live in $KO^{p-7}(T^2,\widetilde w_2)$, where
$\widetilde w_2\in H^2(T^2,\bZ_2)$ is non-zero (see the 
first column in Table \ref{table:twistedKO}). The $Dp$-brane
charges in the type IIB theory with $3$ $O^+$-planes and $1$
$O^-$-plane are classified by $KR_{(+,+,+,-)}^{-2}(S^{8-p,0}\times 
S^{1,1}\times S^{1,1},S^{1,1}\times S^{1,1})$. (See the third column
in Table \ref{table:twistedKO}.\,\footnote{One might expect the need to
  add an additional twisting due to the $B$-field, but as already
  noted, non-trivial $B$-fields do not affect $O$-planes that do not
  wrap the compact directions. The affect of the non-trivial $B$-field
  is already encoded in the sign choice.}) 

In our calculation of $KR^{-j}(M)$ in \cite{Doran:2013sxa}
we used the exact sequence: 
\begin{equation}
\label{eq:KRbycutting}
\cdots \to KO^j \xrightarrow{\rho} 
KO^{j-1}\oplus KO^j\to \tKR^j(M) \to KO^{j+1} \xrightarrow{\rho} KO^{j} \oplus
KO^{j+1} \to \cdots.
\end{equation}
The connecting maps $\rho$ are given by cup product with a class in
$KO^{-1}\cong \bZ_2$ (into the first summand), which turned out to be
non-zero (see Section \ref{sec:KRspecies1}), and a class in
$KO^0\cong \bZ$ (into the second summand), which turned out to be
zero. Note that if the connecting map were trivial then we would
obtain the short exact sequence 
\begin{equation}
0\to KO^{j-1}\oplus KO^j\to \tKR^j(M) \to KO^{j+1}\to 0.
\end{equation}
This would give
\begin{equation}
\tKR^j(M)\cong KO^{j-1}\oplus KO^j\oplus KO^{j+1},
\end{equation}
or
\begin{equation}
KR^j(M)\cong KO^{j-1}\oplus 2KO^j\oplus KO^{j+1},
\end{equation}
since $KR^j(M)\cong\tKR^j(M)\oplus KO^j(M)$. This is just the
$KR$-theory for the type I theory with trivial $B$-field. So
mathematically, we see that the difference in the brane classification
for the type I theory with non-trivial $B$-field from that with
trivial $B$-field comes from the non-triviality of the connecting maps
$\rho$ in equation \eqref{eq:KRbycutting}, and thus must be related to
the twisting (which is $2$-torsion). 

Now let us return to the $D$-brane sources. A lot of information can
be gained by looking at the brane charges for the three groups side by
side; see Table \ref{Table:charges}. 

As noted previously, the BPS spectrum is the same for all three
groups. As an example, consider the BPS $D7$-brane charge. As with all
the other cases, in the type I theory with $B=\frac{1}{2}$ this
corresponds to a wrapped $D9$-brane. In the type IIA theories it
corresponds to a $D8$-brane wrapping the fixed circle, and in the type
IIB theory with $4$ fixed points it corresponds to an unwrapped
$D7$-brane. As before, $2$ half $D7$-branes located at the
$O7^-$-plane can form a $D7$-brane in the bulk, which can be explained
as a $D7$-brane at one of the $O^+$-planes, showing why the BPS
spectrum is unchanged. Note that the only cases where there could be any
possible confusion are the values of $p$ for which there are $2$
sources of BPS charge. This happens for $D4$-
and $D0$-branes. In both cases there are $D(p+1)$-branes wrapping
$1$-cycles in the type IIB theories, where there are $2$ distinct
$1$-cycles to wrap. In the type IIA theories, where there is only
one $1$-cycle that can be wrapped by a BPS brane, we have a wrapped
$D(p+2)$-brane and an unwrapped $Dp$-brane. 

Determining the non-BPS sources is more complicated, but we can draw
some conclusions by comparing the three groups that still need to be
verified by a boundary state analysis. There are only $3$ values of
$p$ for which the $Dp$-brane charge contains torsion; they are
$p=7,6$, and $-1$. 

The $p=7$ and $p=-1$ cases are related by Bott periodicity, so we will
only describe the situation for the non-BPS $D7$-brane charge. Then the
$D(-1)$-brane charge source can be obtained by shifting the degree by
$8$. We will also only describe the situation for the type IIA
theories, since the IIB theory can be obtained following the
$T$-dualities described. There are $3$ sources for non-BPS $D7$-brane
charges in the species $2$ type IIA theories:  an unwrapped
$D7$-brane, a wrapped $D9$-brane, and a $D8$-brane wrapping
$S^{1,1}$. For species $0$ there is one source of non-BPS
$D7$-brane charge. This can correspond to a $D9$-, $D8$-, or
$D7$-brane depending on where in the moduli space we are. The
important feature here is that the unwrapped $D7$-brane is related to
the wrapped $D8$-brane based on the radius of $S_{(+,-)}^{1,1}$. As
noted in the calculation of the $KR$-theory for the species $1$ case
in Section \ref{sec:KRcalcs}, the $KR$-theory for $M$ with the fixed
circle removed gives $KSC$, showing that away from the fixed circle
the species $1$ case should contain the species $0$ charges. Let's
first consider the $D8$-brane wrapping $S^{1,1}$, which appears in both
the species $2$ and $0$ groups. It seems safe to assume that this is a
source for non-BPS $D7$-brane charge for the species $1$ group when
one compact direction is small and the other is large, 
for the same reason that it contributed non-BPS charge in the other
cases. The copy of $S^{1,1}$ it wraps in the species $1$ case is the
circle perpendicular to the fixed circle (for $\tau=e^{i\theta}$ this
is the diagonal, $S_D$, or anti-diagonal, $S_A$). $T$-dualizing both
directions will exchange $S_D$ 
and $S_A$, sending the $D8$-brane wrapping the copy of $S^{1,1}$ in one
IIA theory to a $D8$-brane wrapping the copy of $S^{1,1}$ in the
$T$-dual IIA theory, which also has one large compact direction and one
small one. This shows that if the $D8$-brane wrapping 
$S^{1,1}$ is stable for a species $1$ type IIA theory, it must also be stable
for the doubly $T$-dual IIA theory. Now in the species $0$ case the $D8$-brane
wrapping $S^{1,1}$ and wrapped $D9$-brane are stable in different
regions of the moduli space, so it would not make sense to include a
wrapped $D9$-brane and $D8$-brane wrapping $S^{1,1}$ in the same
region of stability. However, if we include an unwrapped $D7$-brane in
the other $T$-dual type IIA theory, we will have a wrapped
$D9$-brane. Therefore all that is left that the second source can be
is a $D8$-brane wrapping $S_{(+,-)}^{1,1}$. For $\tau=e^{i\theta}$
the copy of $S_{(+,-)}^{1,1}$ that is wrapped is parallel to the fixed
circle, but shifted by a half. It is easy to show, following similar
arguments, that it is not possible to construct a consistent situation
where the $D8$-brane wrapping $S^{1,1}$ is not stable, since unwrapped
$D7$-branes and wrapped $D9$-branes are stable in different
regions. Therefore we 
see that in both the large and small volume limit the non-BPS
$D7$-brane sources are a $D8$-brane wrapping $S^{1,1}$ and a
$D8$-brane wrapping $S_{(+,-)}^{1,1}$. The important feature that leads
to this conclusion is that from the previous $2$ cases we saw that a
$D8$-brane is stable whether or not it wraps $S^{1,1}$ or $S^{1,1}_{(+,-)}$,
unlike the other branes involved.

For the non-BPS $D6$-brane charge there is also a unique possibility for
consistent $T$-dual sources. As noted, away from the fixed circle we
would expect the species $1$ real elliptic curve to contain the
sources of non-BPS brane charges from the species $0$ real elliptic
curves. If this source was a $D7$-brane wrapping $S^{0,2}$, after
performing a double $T$-duality we would get back a wrapped
$D7$-brane. This would not leave any room for the second source of
non-BPS charge since the wrapped $D8$-brane and unwrapped $D6$-branes
are stable in different regions. This implies that the $2$ sources of
non-BPS $D6$-brane charge should be an unwrapped $D6$-brane and
a wrapped $D8$-brane. 

As one last interesting note on the non-BPS brane charges, consider
the source for non-BPS $D3$- and $D2$-brane charge that appears in the
species $0$ group, but not in the species $1$ or $2$ groups. The
source in the species $0$ IIA theory with an $O8^+$- and an $O8^-$-plane
is an unwrapped $D3$-brane located at the $O8^-$-plane. This
corresponds to a $D4$-brane stretched between the $2$ $O7^-$-planes in
the $T$-dual type IIB theory with $4$ fixed points. Since the type IIB
theories with $4$ fixed points for the species $1$ and $2$ groups do
not have $2$ $O7^-$-planes, this charge cannot exist in these theories
and does not appear in their $K$-theory spectra.

\section{Conclusion}
\label{sec:concl}
Let us summarize what we have accomplished in this paper.  We have
studied \emph{all} orientifold string theories on space-times of the
form $E\times \bR^{8,0}$, where $E$ is an elliptic curve with
holomorphic or anti-holomorphic involution. These are quite natural
spacetimes to consider since elliptic curves are the only compact
Calabi-Yau manifolds of complex dimension $1$, and only
holomorphic or anti-holomorphic involutions will preserve supersymmetry.
These theories divide into
three groups, and all the theories within each group are related to
one another by sequences of $T$-dualities.  For each theory, there is a
corresponding twist (given by the sign choice on the $O$-planes and/or
the B-field), and the twisted $KR$-theory classifies the
D-brane charges.  We determine not only the charge groups but also the
precise brane content for each theory.  To the best of our knowledge,
the brane content of the type I theory without vector structure was
not previously known.

It is worth pointing out a few key points:
\begin{enumerate}
\item The torsion-free part of the $KR$-groups classifies the BPS
  spectrum and does not depend on the twists.
  Twisting only affects the $2$-torsion in the $KR$-groups, not
  the torsion-free part of the groups.
\item Each $T$-duality grouping includes precisely one IIB theory with
  four $O$-planes.  The signs of these $O$-planes can be read off from
  the Legendre normal form of the corresponding real elliptic curve
  with involution, and are reflected in the uniformization of the
  elliptic curve via Jacobi functions.
\item Each $T$-duality grouping also includes a unique variant of type I
  string theory, or in other words, a IIB theory where the holomorphic
  involution on $E$ is either trivial or free.  Possibilities for this
  theory are the conventional type I theory, the type I theory without
  vector structure, and the type \~I theory.
\item A full stability analysis of the various classes of branes still
  remains to be done in some cases, but what we have done here is a first
  approximation based on understanding of theories compactified on a
  circle. For the ``type I theory without vector structure,'' our
  understanding is already complete.
\item The $T$-duality groupings can be understood from either purely
  mathematical or purely physical points of view. It is quite dramatic
  that the calculations of the twisted $KR$-groups (which is pure
  algebraic topology) and the classifications via Legendre normal
  forms (which is pure algebraic/analytic geometry) both confirm what
  had been conjectured by physicists many years ago.
\end{enumerate}

\bibliographystyle{hplain}
\bibliography{T2}
\end{document}